\documentclass[11pt,onecolumn] {IEEEtran}
\usepackage{amsmath}
\usepackage{amsthm}
\usepackage{amssymb}
\usepackage{graphicx}
\usepackage{color}
\usepackage{multicol}
\usepackage{epstopdf} 
\usepackage[ruled,vlined]{algorithm2e}

\ifCLASSINFOpdf
 \else
 \fi

\newtheorem{corollary}{Corollary}
\newtheorem{theorem}{Theorem}
\newtheorem{lemma}{Lemma}
\newtheorem{remark}{Remark}
\newtheorem{definition}{Definition}

\newcommand{\ba}{\mathbf{a}}
\newcommand{\bA}{\mathbf{A}}

\newcommand{\bg}{\mathbf{g}}
\newcommand{\bG}{\mathbf{G}}

\newcommand{\bx}{\mathbf{x}}

\newcommand{\by}{\mathbf{y}}

\newcommand{\bz}{\mathbf{z}}

\begin{document}
\title{Signal Reconstruction from Interferometric Measurements under Sensing Constraints}

\author{Davood~Mardani$^\dagger$, George~K.~Atia$^\dagger$ and Ayman F. Abouraddy$^\ddagger$
\thanks{This material is based upon work supported by the U.S. Office of Naval Research contract N00014-14-1-0260, NSF CAREER Award CCF-1552497, and NSF grant No. CCF-1320547.

$\dagger$ D. Mardani and G. K. Atia are with the Department of Electrical and Computer Engineering, University of Central Florida, e-mails: d.mardani@knights.ucf.edu, george.atia@ucf.edu.

$\ddagger$ A. F. Abouraddy is with CREOL, The College of Optics \& Photonics, University of Central Florida, email: raddy@creol.ucf.edu.}
}


\maketitle

\begin{abstract}
This paper develops a unifying framework for signal reconstruction from interferometric measurements that is broadly applicable to various applications of interferometry. In this framework, the problem of signal reconstruction in interferometry amounts to one of basis analysis. Its applicability is shown to extend beyond conventional temporal interferometry -- which leverages the relative delay between the two arms of an interferometer -- to arbitrary degrees of freedom of the input signal. This allows for reconstruction of signals supported in other domains (e.g., spatial) with no modification to the underlying structure except for replacing the standard temporal delay with a generalized delay, that is, a practically realizable unitary transformation for which the basis elements are eigenfunctions. Under the proposed model, the interferometric measurements are shown to be linear in the basis coefficients, thereby enabling efficient and fast recovery of the desired information. While the corresponding linear transformation has only a limited number of degrees of freedom set by the structure of the interferometer giving rise to a highly constrained sensing structure, we show that the problem of signal recovery from such measurements can still be carried out compressively. This signifies significant reduction in sample complexity without introducing any additional randomization as is typically done in prior work leveraging compressive sensing techniques. We provide performance guarantees under constrained sensing by proving that the transformation satisfies sufficient conditions for successful reconstruction of sparse signals using concentration arguments. We showcase the effectiveness of the proposed approach using simulation results, as well as actual experimental results in the context of optical modal analysis of spatial beams.
\end{abstract}

\section{Introduction}\label{sec:intro} 
Interferometry is a measurement strategy that is widely used across all the physical sciences, with applications ranging from astronomy and radio interferometry \cite{Pan2017_radio, Thompson2007_radiobook}, to remote sensing and Interferometric Synthetic Aperture Radar (InSAR) \cite{Rechards2007_SAR, Ferretti2007_SARbook, SARsignalmagazine}, optics and photonics \cite{Born1999_principleoptics, Monnier2003_optics, Abouraddy11OL,Kurien:14}, signal processing and communications \cite{Angleestimate, localization, imagesignalmagazine}, optical encryption \cite{Li15SCIrep,Rawat15AP}, and bio-imaging \cite{OCTprinciple, baha_OCT}. Underlying the utility of interferometry in all these fields is the fundamental principle of superposition of linear waves, which applies to optical, radio-frequency, and acoustic waves, among other physical realizations. By judiciously superposing two versions of a wave, their interference may reveal sought-after information, typical about a sample or a medium that one of the waves scattered from. The interferometer in which the superposition takes place may be an instrument implemented using electrical and optical components (e.g., a Michelson interferometer in Optical Coherence Tomography (OCT) \cite{OCTprinciple}), or simply a physical medium (e.g., the atmosphere in the case of localization in wireless sensor networks \cite{localization, localizationexample}).

Common to all such problems are interferometric measurements, so-called \emph{interferograms}, obtained by acquiring the energy of the superposition of the two waves or signals while some parameter is swept \cite{interferometrystellar, Born1999_principleoptics}. The interferogram typically assumes values related to the auto-correlation and cross-correlation of the signals in the interferometer, which depend on the characteristics of its arms (e.g., their physical lengths in temporal interferometry). For example, in time domain OCT one can acquire several interferometric measurements by sweeping the time delay in one of the interferometer arms \cite{OCT:book}.

This paper provides a fresh perspective on the problem of signal reconstruction from interferometric measurements. In particular, we propose a unifying framework in which the problem of signal recovery from interferograms amounts to basis analysis in some appropriate Hilbert space. For example, in temporal interferometry this space is the span of the set of complex harmonics. We show here how this concept generalizes to other bases enabling interferometry in a variety of bases related to any degree (or degrees) of freedom of the wave. In our model, we show that the interferometric measurements are linear in the information of interest (embedded in the expansion coefficients in the basis), thereby enabling more efficient and faster algorithms for recovery of the desired information. Although the measurement model has only few degrees of freedom set by the structure of the interferometer, we show both analytically and experimentally that the problem of basis analysis is amenable to compressive reconstruction whereby significant reductions in sample and computational complexity can be achieved given prevalent sparse representations in that basis. We show that the linear transformation of the measurement model satisfies sufficient conditions for successful reconstruction even under the sensing constraints set by the limited degrees of freedom of the interferometer.

{\color{black} We emphasize the generality of our approach by first giving an overview of sample applications of interferometry that can be studied in light of the proposed framework as we further elaborate in Section \ref{ssec:applications} and Section \ref{sec:generalmodel}. 
OCT is a non-invasive and contact-free optical imaging method which provides high-resolution depth and transversal images from different layers of a sample object \cite{OCT_science,OCTprinciple}, and is a heavily used bio-imaging technique in ophthalmology to capture high resolution cross-sectional images of the retina \cite{OCT_retina}. 
In OCT, a low-coherence source emits a light beam that scatters off a sample object such as living tissue as shown in Fig. \ref{fig:combination_apps}(a). The scattered light is then combined with a delayed version of the input beam to reveal the depth information of the object \cite{OCT:book}. In this example, the path which has the sample object corresponds to one arm of the interferometer whose reflectivity indices at the different layers are of interest. In the proposed framework, we show that the reflectivity indices appear in expansion coefficients related to the interferometric measurements. A second example is that of optical modal analysis \cite{Abouraddy11OL,Abouraddy12OL,Mardani15OE}, in which measurements collected using an optical interferometer -- such as the Mach–-Zehnder interferometer of Fig. \ref{fig:combination_apps}(b) -- are used to reveal the modal content of an optical beam. This example is studied in detail in Section \ref{sec:results}. Another example pertains to localization in wireless sensor networks. As shown in Fig. \ref{fig:combination_apps}(c), to determine position, a node receives two signals with two different frequencies from two adjacent anchor nodes. By synchronizing the receiver and transmitters, the delay of each path defines the distance between the nodes. Hence, the position information is embedded in the energy of samples of the combined signal \cite{localization}.}

\subsection{Related work}\label{works}
\subsubsection{Signal recovery from intensity measurements}
OCT \cite{OCTprinciple}, angle-of-arrival estimation \cite{Angleestimate}, and localization in wireless sensor networks \cite{localization} are examples of applications of interferometry in which the information of interest relates to the structure of the interferometer -- for example the path delay or equivalently the length of each arm. 
Another class of problems in interferometry are those in which the interferometric measurements are used to reveal the properties of the input signal, such as in the above-mentioned optical modal analysis example to reveal the modal content of light beams \cite{Abouraddy11OL,Abouraddy12OL}. In high-rate optical communications, the data modulates different temporal or spatial modes of a light beam \cite{Wang12NP}, then modal analysis can be used to decode the data at the receiver's side. 

Phase retrieval refers to another important class of problems with broad interest in which one seeks to recover a complex signal given only amplitude or intensity measurements \cite{Fienup_phase, Eldar_phase, Pal_ICASSP, Pal_GLOBAL}. This is to be contrasted to the interferometric setup we consider herein where signal recovery is intended from correlation-type measurements. In \cite{Phaselift}, the authors devise an approach termed PhaseLift to recover the signal of interest by searching for a rank-one solution of a formulated matrix recovery problem. Leveraging the quadratic measurement model, the idea is to reformulate the problem of signal recovery from quadratic constraints in terms of rank minimization via semidefinite programming. The interferometry problem considered in this paper is not amenable to similar formulations given the constrained sensing structure imposed by the architecture of the interferometer and the underlying measurement model.

\subsubsection{Compressive techniques}
In dealing with the practical limitations of sensing systems, prior work on compressive signal recovery has relied on introducing additional (non-native) hardware components to emulate randomization. For example, for optical field estimation and imaging, the field is projected onto a sequence of random masks inserted along the field path in \cite{Mirhosseini14PRL, Howland14PRL}. Similarly, the single pixel camera \cite{Duarte08SPM,Martinez-Leon:17,Clemente:13} uses a time-varying random mask to acquire random projections of a scene instead of directly collecting the pixels/voxels using a large size detector. Random masks in the form of a Digital Micromirror Device (DMD) (an array of millions of individually addressable and tiltable mirror-pixels) are also utilized in optical encryption for secure communication in optical networks to compress the encrypted data prior to transmission \cite{Li15SCIrep,Rawat15AP}. In sharp contrast, this paper proposes a compressive method for signal reconstruction from interferometric measurements acquired \emph{without modifying the underlying sensing architecture.}

Our work also departs from related literature on compressive signal recovery in interferometry such as \cite{baha_OCT,Liu:10} in which compression is restricted to the recovery step post data collection. By contrast, our approach uses compressive data acquisition/sensing so that we collect less interferometric measurements in the first place for subsequent reconstruction. This enables savings in both the data acquisition and signal recovery steps.  

\subsection{Contributions}\label{contribution}
In the following, we summarize the main contributions of this paper.
\smallbreak
\noindent$\bullet$ \emph{A unifying framework for temporal interferometry (basis analysis)}: 
We propose a unifying framework in which every problem of temporal interferometry is mapped to one of basis analysis -- the information of interest is embedded in the basis coefficients. The interferogram admits an explicit linear representation of known structure in the desired information about the input signal or some sample object\footnote{Common interferometry models express the measurements at the level of correlation terms which are not explicit in the information of interest.}. This problem-independent generalized linear model enables faster algorithms for signal reconstruction from interferograms. 
\begin{figure*}[t!]
	\centering
	\includegraphics[scale=.9]{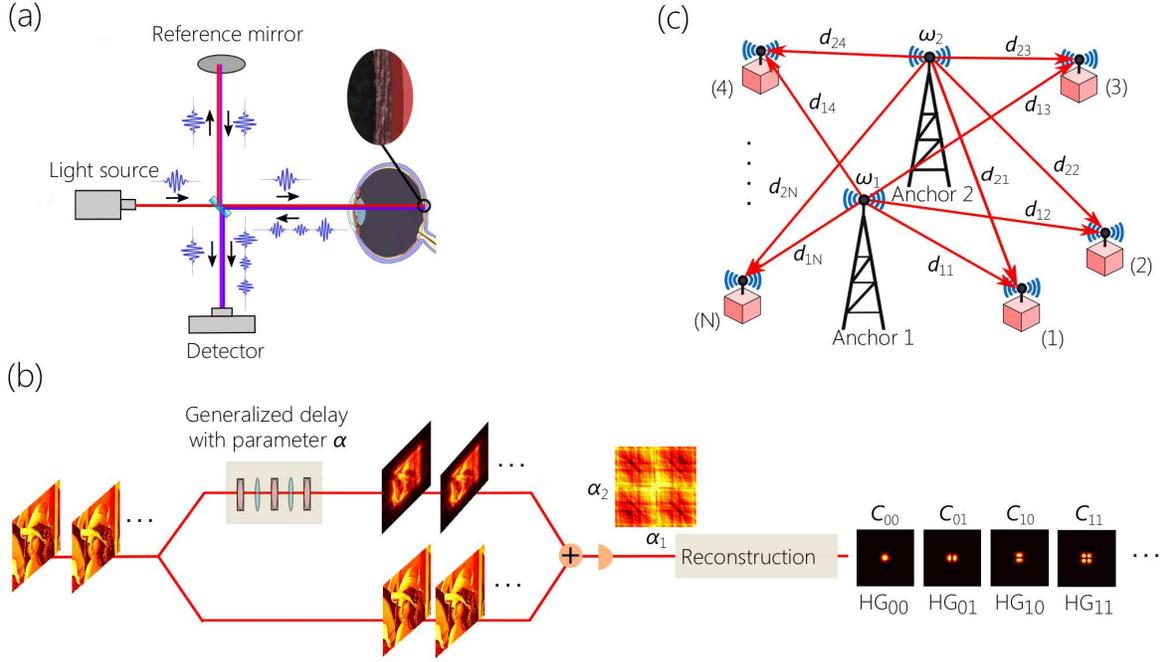}
	\vspace{-10pt}\caption{(a) Schematic of the OCT implementation using a Michelson interferometer. The reference beam is combined with the beam reflected from different layers of a sample object. The intensity measurements collected by the detector are used to retrieve reflectivity indices of different layers of the sample.
    (b) Two-dimensional modal analysis using the generalized interferometry approach in Section \ref{sec:generalmodel} based on Hermite-Gaussian modes. The intensity profile of a scene is passed through two cascaded fractional Fourier transforms with parameters $\alpha_1$ and $\alpha_2\!\in\![0,2\pi]$. The output is superposed with the input beam to produce an interferogram from which the modal content of the beam is revealed using an appropriate reconstruction algorithm. (c) Topology of a wireless sensor network used for node localization. Two anchor nodes transmit two sinusoids with different frequencies. The distances between receiver nodes and the anchor nodes are determined intensity measurements of the superposed signals.  }
	\label{fig:combination_apps}
\end{figure*}

\noindent$\bullet$ \emph{Generalization beyond the temporal degree of freedom}: The basis analysis framework is shown to generalize to other degrees of freedom beyond the temporal, such as the spatial parameters. The key idea underlying our ability to make such generalization is replacing the standard time delay commonly used in temporal interferometry with suitable unitary transformations (for which the elements of the basis are eigenfunctions) to allow for analysis in arbitrary bases for the other degrees of freedom, hence the designation `generalized delays'. 
This game-changing result allows us to perform (generalized) interferometry in arbitrary degrees of freedom with no modification to the interferometer structure, opening up new venues for signal recovery from interferograms. 

\noindent$\bullet$ \emph{Compressive interferometry}: The linear model for basis analysis enables compressive techniques for data acquisition and signal recovery when the signals admit sparse representations in some appropriate basis. In sharp contrast to prior work, we do not modify the underlying interferometer structure. Therefore, compressive signal reconstruction is carried out under sensing constraints set by the limited degrees of freedom of the interferometer. Despite the constrained sensing structure, we show using concentration arguments that the linear transformation satisfies sufficient conditions for successful sparse recovery, such as the Restricted Isometry Property (RIP) \cite{Candes08CR}, and the isotropy and incoherence properties for random ensembles \cite{candesRIPless}. As such, we are able to provide performance guarantees for signal recovery under sensing constraints using native interferometric setups. Furthermore, we demonstrate the proposed compressive generalized interferometry approach in the context of optical modal analysis in a Hermite-Gaussian Basis by realizing an actual generalized delay in the spatial degree of freedom  -- in this case, a fractional Fourier transform -- both in numerical simulations and laboratory experiments.   
\\
\\
\noindent\textbf{Paper organization:} The rest of the paper is organized as follows. In Section \ref{sec:interferometry}, we propose the basis analysis framework for two-path temporal interferometry, and show how various examples from different applications are mapped to the corresponding generalized linear measurement model. In Section \ref{sec:generalmodel}, we introduce the concept of generalized interferometry where we generalize this framework to arbitrary degrees of freedom of the input signal. In Section \ref{sec: CS_practical constraint}, we develop a compressive approach for reconstruction of sparse signals from interferometric measurements under sensing constraints and establish performance guarantees using mathematical analysis. In Section \ref{sec:results}, we provide simulation and experimental results in the context of optical modal analysis demonstrating the proposed compressive approach in generalized interferometry. 

\textit{Notation:} The superscripts $^{\mathrm{T}}$ and $^{\mathrm{H}}$ denote the transpose and the conjugate transpose operators, respectively. The bracket notation $[a_{n}]$ with index variable $n$ denotes a column vector indexed by $n$. Similarly, the notation $[a_{nm}]$ denotes a matrix with rows and columns indexed by $n$ and $m$, respectively. An orthonormal basis with elements $\phi_n$ indexed by $n$ for a given degree of freedom $x$ is denoted $\{\phi_n(x)\}$.

\section{A unifying framework for interferometry}
\label{sec:interferometry}

\subsection{Interferogram model}\label{ssec:intro_interf}

A generic interferometric configuration is depicted schematically in Fig.~\ref{fig:interferometer}. An input signal or optical field $\psi(t)$, where $t$ corresponds to time, is divided into two paths (or interferometer arms), whereupon two new versions $\psi_{1}(t;\tau)$ and $\psi_{2}(t)$ are created and combined to produce a superposed signal,
\begin{equation}\label{eq:summation}
\psi_{\mathrm{s}}(t;\tau)=\psi_{1}(t;\tau)+\psi_{2}(t).
\end{equation}
The first arm (referred to as the `reference' arm) has an impulse response $h_1(t;\tau)\!=\!\delta(t-\tau)$ where $\tau$ is a temporal delay, and the second arm (the `sample' arm) has an impulse response $h_2(t)$. An `interferogram' is traced by scanning over the values of $\tau$ and recording the energy of the superposed signal $I(\tau)\!=\!\langle|\psi_{\mathrm{s}}(t;\tau)|^{2}\rangle$, where $\langle\cdot\rangle$ corresponds to an integration over time. The interferogram is thus given by
\begin{equation}\label{eq:interferogram3}
I(\tau)\!=\!I_{1}+I_{2}\!+\!2|I_{12}(\tau)|\cos(\theta_{12}(\tau)).
\end{equation}
The first two terms on the right hand side of (\ref{eq:interferogram3}) represent the auto-correlation of the signals produced in each arm of the interferometer (the total energy of the signal in each arm),
\begin{equation}\label{eq:def_auto_corr}
I_{1}\!\triangleq\!\langle\psi_{1}(t;\tau)\psi_{1}^{*}(t;\tau)\rangle,\,
I_{2}\!\triangleq\!\langle\psi_{2}(t)\psi_{2}^{*}(t)\rangle,
\end{equation}
whereas the third term captures their cross-correlation,
\begin{equation}\label{eq:cross_corr}
I_{12}(\tau)\triangleq\langle\psi_{1}(t;\tau)\psi_{2}^{*}(t)\rangle,
\end{equation}
and $\theta_{12}(\tau)$ is the phase of $I_{12}(\tau)$. Interferometric measurements are collected by sampling the delay $\tau$, and the sought-after information about the input signal or the interferometer arms is typically embedded in the cross-correlation term.

\begin{figure}[b]
	\centering
	\includegraphics[scale=.9]{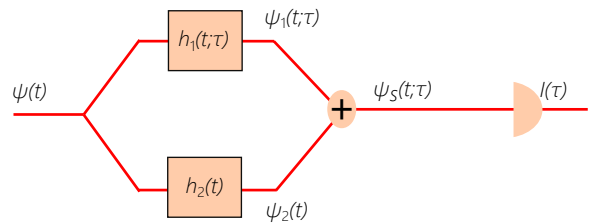}
	\caption{Schematic for a general two-path interferometer. The output signal of the reference arm is defined by the temporal delay $\tau$.}
	\label{fig:interferometer}
\end{figure}

Although (\ref{eq:interferogram3}) provides a general model for interferometry that is commonly used, it unfortunately does not show explicitly how the correlation term relates to the information of interest. In the next subsection, we modify this traditional model of temporal interferometry (Fig. \ref{fig:interferometer}) such that it is mapped to a problem of basis analysis. It will be shown that the interferogram is linearly related to the information of interest, whether this information pertains to the input signal or to the `sample'.

\subsection{Basis analysis -- linear measurement model}
\label{sec:basis_analysis}
We expand the input signal $\psi(t)$ in terms of harmonics or complex exponentials $\{e^{j\omega t}\}$, where $\omega$ is the angular frequency (i.e., the Fourier basis), such that $\psi(t)\!=\!\tfrac{1}{2\pi}\int_{-\infty}^{+\infty}\Psi(\omega)e^{j\omega t}d\omega$, where $\Psi(\omega)$ is the Fourier transform (FT). Hereon, we focus our attention on discrete bases by the mere fact that the data collected and the information retrieved is always represented discretely. In this case, the input signal is represented as $\psi(t)\!=\!\sum_{n=1}^{\infty}c_{n}e^{j\omega_{n}t}$ using the orthogonal discrete harmonics $\{e^{j\omega_{n}t}\}$ for some complex coefficients $c_n, n=1,2,\ldots$. Because of the discrete basis, the signal is periodic in time, so that all integrals over time extend over this period. The delay introduces a phase factor $e^{-j\omega_n\tau}$ to the coefficient $c_n$ that is linear in $\tau$ and the modal `index' $\omega_n$, 
\begin{equation}\label{eq:arm1_output}
\psi_{1}(t;\tau)=\psi(t-\tau)=\sum_{n=1}^{\infty}c_{n}e^{j\omega_{n}t}e^{-j\omega_{n}\tau}.
\end{equation}
This fact will be utilized subsequently when introducing the notion of a `generalized delay' for non-temporal degrees of freedom. 

Modeling the sample arm of the interferometer as a linear time-invariant system $h_2(t)$, its output will be
\begin{equation}\label{eq:arm2_output}
\psi_{2}(t)\!=\!\sum_{n=1}^{\infty}d_{n}e^{j\omega_{n}t}, 
\end{equation}
where $d_{n}\!=\!c_{n}H_{2}(\omega_{n})$, $n\!=\!1,2,\ldots$, and $H_{2}$ is the Fourier transform of $h_{2}$. From (\ref{eq:interferogram3}) and the orthogonality of the complex harmonics, the interferogram becomes
\begin{equation}\label{eq:interferogram_exponential}
I(\tau)\!=\!\sum_{n=1}^{\infty}|c_{n}|^{2}\!+\!\sum_{n=1}^{\infty}|d_{n}|^{2}\!+\!2\sum_{n=1}^{\infty}|c_{n}||d_{n}|\cos(\omega_{n}\tau\!+\!\theta_{n}),
\end{equation}
where $\theta_{n}$ is the phase of $c_{n}d_{n}^{*}$. As the first two terms do not depend on $\tau$, we define the interferometric measurements as,
\begin{equation}\label{eq:define_linearmeasurements}
\begin{split}
y(\tau)\triangleq\frac{1}{2}\left\{I(\tau)-\sum_{n=1}^{\infty}\left(|c_{n}|^{2}+|d_{n}|^{2}\right)\right\}
=\sum_{n=1}^{\infty}|c_{n}||d_{n}|\cos(\omega_{n}\tau+\theta_{n}).
\end{split}
\end{equation}
To collect $M$ interferometric measurements, we sample $M$ values $\tau_{m}$, $m\!=\!1,\ldots,M$, of the delay $\tau$.

\begin{figure*}[t]
	\centering
	\includegraphics[scale=.9]{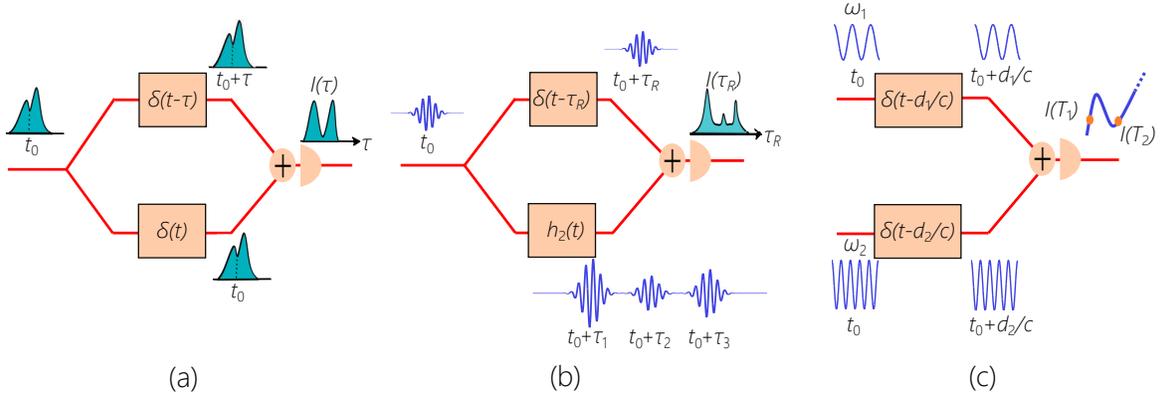}
	\vspace{-10pt}\caption{(a) Block diagram of the interferometry-based modal analysis. (b) Block diagram of OCT where the reference arm is modeled by a delay block, and the sample arm is modeled by an LTI system.(c) Block diagram of an interferometric based localization technique. }
	\label{fig:combination_block}
\end{figure*}

In the framework of temporal interferometry described here and in all subsequent modalities, only a finite number of coefficients $c_{n}$ and $d_{n}$ are of interest or even accessible by the acquisition systems. We thus introduce at this point a finite dimensionality $N$ for the harmonic basis, $\{e^{j\omega_{n}t}\}_{n=1}^{N}$. Considering (\ref{eq:define_linearmeasurements}), we therefore obtain a linear model for the interferometric measurements in the time domain 
\begin{equation}\label{eq:interferogram_linearmodel}
\mathbf{y}=\mathbf{A}\mathbf{x},
\end{equation}
where the $M\!\times\!1$ measurement vector $\mathbf{y}$ contains the interferometric measurements $y(\tau_{m})$, the $2N\!\times\!1$ information vector $\mathbf{x}\!=\![\mathbf{x}^{\mathrm{T}}_{1} \,\,\,\mathbf{x}^{\mathrm{T}}_{2}]^{\mathrm{T}}$ consists of the two vectors $\mathbf{x}_{1}\!=\![|c_{n}||d_{n}|\cos(\theta_{n})]$ and $\mathbf{x}_{2}\!=\![|c_{n}||d_{n}|\sin(\theta_{n})]$, and the $M\!\times\!2N$ matrix $\mathbf{A}\!=\![\mathbf{A}_{1}\, \,\,\mathbf{A}_{2}]$ is a block matrix with $\mathbf{A}_{1}\!=\![\cos(\omega_n\tau_m)]$, and $\mathbf{A}_{2}\!=\![-\sin(\omega_n\tau_m)]$, with $m\!=\!1,2,\ldots,M$ and $n\!=\!1,2,\ldots,N$. 

The goal here is to recover some properties of the input signal (the coefficients $c_{n}$) or of the sample (the coefficients $d_{n}$) from the interferogram. The interferogram model (\ref{eq:interferogram_linearmodel}) offers an immediate advantage. We have reduced \textit{every} problem of temporal interferometry with the configuration in Fig.~\ref{fig:interferometer} to one of basis analysis. This is a unifying problem-independent framework in which the interferometric measurements admit a linear representation in terms of a matrix $\mathbf{A}$ of known structure, which enables more efficient approaches to information recovery. Next, we demonstrate two important applications mapped to this generalized linear model. In the first example we seek to reconstruct the input signal while in the second we seek to identify $h_2(t)$.

\subsection{Examples}\label{ssec:applications}

\smallbreak
\noindent$\bullet$\,\,\textbf{\textcolor{black}{Interferometry-based optical spectroscopy:}}
This example concerns analyzing a pulsed optical field into its constituent temporal modes via the interferometric configuration in Fig.~\ref{fig:combination_block}(a). Hence, the goal is to reconstruct the modal coefficients $c_n,n=1,\ldots,N$ of an optical field represented in a finite basis $\psi(t)\!=\!\sum_{n=1}^{N}c_{n}e^{j\omega_{n}t}$. The field $\psi(t)$ enters a two-path optical interferometer such as a Mach-Zehnder interferometer (MZI), and a delay $\tau$ is swept. The input field passes through the sample arm without undergoing any change, i.e., this is a special case of the general model in \ref{sec:basis_analysis} where $h_2(t)\!=\!\delta(t)$, thus $\psi_{2}(t)\!=\!\psi(t)$. Normalizing the energy to unity $\sum_{n=1}^{N}|c_{n}|^{2}\!=\!1$, the interferometric measurements are
\begin{equation}\label{eq:modal_linearmodel}
y(\tau_m)=\sum_{n=1}^{N}|c_n|^2\cos(\omega_n \tau_m), \,\,\,m=1,2,...,M.
\end{equation}
Hence, given these $M$ measurements, the $M\!\times\!1$ measurement vector $\mathbf{y}$ with entries $y(\tau_m)$, fits the linear model in 
(\ref{eq:interferogram_linearmodel}) for a modal coefficient vector $\mathbf{x}\!=\![|c_{1}|^2 |c_{2}|^2\ldots|c_{N}|^{2}]^{\mathrm{T}}$ and an $M\!\times\!N$  matrix $\mathbf{A}\!=\![\cos(\omega_{n} \tau_{m})]$. Revealing the spectrum of the optical field thus amounts to solving a system of linear equations.

\smallbreak
\noindent$\bullet\,\,$\textbf{Time-Domain OCT (TD-OCT):} OCT makes use of a low-coherence (large-bandwidth) optical source in a two-path interferometer such as a Michelson interferometer \cite{Born1999_principleoptics}, as illustrated in Fig.~\ref{fig:combination_apps}(a). If the spectrum of this source is $s(\omega)$, we discretize it and obtain the coefficients $c_{n}\!=\!s(\omega_{n})$. A layered sample is placed in the sample arm, a delay is swept in the reference arm, and the time-averaged energy of the superposed signal is recorded for each delay to reconstruct the layered sample. Hence, this is an example of interferometry where we seek to recover information about the sample impulse response $h_2(t)$. We model the (typically reflective) layered sample by a linear time invariant impulse response $h_2(t)\!=\!\sum_{\ell=1}^Lr_{\ell}\delta(t-T_{\ell})$, which is parametrized by the round-trip time $T_{\ell}$ for the field to travel from the $\ell^{\text{th}}$ sample layer to the sample surface, and $r_{\ell}$ is the field reflectivity of the $\ell^{\mathrm{th}}$ layer. Because the reflection from typical biological samples is very weak, terms that are higher than first order in $r_{\ell}$ are usually ignored. Accordingly, the output from the sample arm
is characterized by the coefficients $d_n\!=\!c_{n}\sum_{\ell=1}^{L}r_{\ell}e^{-j\omega_n T_{\ell}}, n=1,2,\ldots,N$. Assuming that the source is well-characterized (i.e., the coefficients $c_{n}$ are known), then our linear model retrieves the coefficients $d_{n}$. From these, both the reflectivity of the layers and their depths with respect to the sample surface can be reconstructed.

\begin{remark}
\label{rem:loc}
A related example is that of localization in wireless sensor networks shown in Fig.~\ref{fig:combination_apps}(c). Consider two anchor nodes each transmitting a sinusoid with distinct frequencies $\omega_1$ and $\omega_2$. The transmitted signals $\psi_k(t)\!=\!a_{k}e^{j\omega_{k}t},\: k=1,2$, superpose at the receiver to produce the signal $\psi_{\mathrm{s}}(t)\!=\!\psi_{1}(t-d_{1}/c)+\psi_{2}(t-d_{2}/c)$, where $d_1$ and $d_2$ are the distances between the receiver and the anchor nodes, and $c$ is the speed of light in vacuum. In contrast to standard interferometry, in this case we have no control over the relative delay of the two paths. Instead, the received signal is sampled at different time instants to recover the distances. Although this problem cannot be viewed as one of basis analysis, by sampling $M$ points we again obtain linear measurements $\by\!=\!\bA'\bx$, where $\mathbf{x}=\begin{bmatrix}|a_1||a_2|\cos(\frac{\omega_2d_2-\omega_1d_1}{c}) ~~~ |a_1||a_2|\sin(\frac{\omega_2d_2-\omega_1d_1}{c})
\end{bmatrix}^T$ and the matrix $\mathbf{A}'$ has dimensions $M\!\times\!2$, with the entries in each row being $\cos((\omega_1-\omega_2)T_{m})$ and $-\sin((\omega_1-\omega_2)T_{m})$. This can be easily generalized to multiple receiving nodes. 
\end{remark}

\begin{remark}
The applicability of this two-path interferometry framework is by no means restricted to deterministic (periodic) signals, but applies naturally to stochastic signals as well, by virtue of the linearity inherent in the superposition of fields at the output. The only modification required would be in (\ref{eq:cross_corr}), where the time-average is replaced by an expectation over random field realizations. In the case of ergodic stochastic signals, this expectation can be carried out by averaging over a period of time larger than the `coherence time' that is proportional to the inverse of the bandwidth of the power spectral density, as known from the Wiener-Khinchtine theorem \cite{baha}.    
\end{remark}

\section{Generalized interferometry: Hilbert space analyzers}\label{sec:generalmodel}

In the previous section, we proposed a unifying model for temporal two-path interferometry. In this modality, we have shown that the desired information (the input signal or the sample) appears in the harmonic expansion coefficients of a linear measurement model. It turns out that the framework developed can be generalized to arbitrary degrees of freedom of the input signal beyond the temporal, such as the spatial parameters. Underlying this generalization is the notion of `generalized delay', which replaces the standard temporal delay $\tau$ in Section~\ref{sec:interferometry} to allow for analysis in arbitrary bases for the other degrees of freedom. This is our second main contribution, which is the primary focus of this section on generalized interferometry.

In temporal interferometry, we have represented the input signal as a finite discrete superposition of time-frequency complex exponentials. The delayed output signal of the reference arm in (\ref{eq:arm1_output}) is obtained by passing the input signal through a temporal delay modeled as an LTI system with an impulse response $h_1(t;\tau)\!=\!\delta(t-\tau)$. Equivalently, applying a delay amounts to applying a linear phase factor $e^{-j\omega_n \tau},\,n=1,2,...,N$ to the basis harmonics $e^{j\omega_n t},\,n=1,2,...,N$. In other words, the harmonics $e^{j\omega_nt}$ are the \emph{eigenfunctions} of the delay system $h_1(t;\tau)$ with \emph{eigenvalues} $\lambda_n=e^{-j\omega_n\tau}$.

In moving to other degrees of freedom but maintaining the overall interferometric structure, we must replace the temporal delay with an appropriate `generalized delay'. The signal in this case is an element in a Hilbert space spanned by an orthonormal basis ${\{\phi_n(x)\}}$ with respect to an arbitrary variable $x\in \mathbb{R}$ (e.g. space, angle, etc). As before, we represent the input signal or light field $\psi(x)$ as a superposition of the basis elements, $\psi(x)\!=\!\sum_{n=1}^{N}c_n\phi_n(x)$, where the $c_{n}$'s are the basis coefficients. In this setting, we take the generalized delay $\alpha$ (potentially multi-dimensional) -- represented by an impulse response $h_{1}(x;\alpha)$ -- to be the unitary linear system whose impact on the signal is analogous to that of the temporal delay in (\ref{eq:arm1_output}). In other words, the eigenfunctions of the transformation $h_{1}(x;\alpha)$ must be the Hilbert-space basis $\{\phi_{n}(x)\}$ with eigenvalues of the form $e^{-jn\alpha}$. We refer to $\alpha$ hereon as the generalized delay parameter. With these features taken into consideration, the delay operator in the Hilbert-space basis takes on a diagonal representation,
\begin{equation}
h_{1}(x,x';\alpha)=\sum_{n=1}^{N}e^{-j n\alpha}\phi_{n}(x)\phi_{n}^{*}(x').
\end{equation}
The structure of this operator has several salutary properties that justify calling it a generalized delay. It is additive in the delay parameter $\int\!dx'h_{1}(x,x';\alpha)h_{1}(x',x'';\beta)=h_{1}(x,x'';\alpha+\beta)$; its inverse is the same operator but with a delay parameter $-\alpha$; and $h_{1}(x,x';0)$ is the identity. It has been shown that this structure corresponds in general to fractional transforms. For example, when the basis $\{\phi_{n}(x)\}$ is that of Hermite-Gaussian function, $h_{1}$ corresponds to the fractional Fourier transform (frFT) \cite{NAMIAS80JAM1,Pei07TSP}; when the basis is that of Laguerre-Gaussian functions, $h_{1}$ corresponds to a fractional Hankel transform, etc \cite{NAMIAS80JAM2,Abouraddy12OL}.

As such, the response of this `delay' to the input $\phi_n(x)$ is $\int_{-\infty}^{+\infty}\phi_n(x')h_{1}(x,x';\alpha)dx'$ will be $e^{-jn\alpha}\phi_n(x)$. Thus, a signal $\psi(x)$ after being `delayed' takes the form 
\begin{equation}\label{eq:filter output}
\psi_1(x;\alpha)\!=\!\int_{-\infty}^{+\infty}\!\!\!\psi(x)h_{1}(x,x';\alpha)dx'\!=\!\sum_{n=1}^N c_n e^{-jn\alpha}{\phi}_n(x).
\end{equation}
This idea underlies our approach to conduct interferometry in arbitrary bases related to other degrees of freedom. As pointed out earlier, in this paper we focus on signals in finite-dimensional Hilbert spaces since in practice only few basis elements contribute to the actual signal or can be accessed. In the two examples described in the previous Section, spectroscopy and OCT, only a few harmonics or reflection layers contribute to the signal analyzed. For spatial signals, the highest-order basis coefficient is limited by the size of the optical aperture.

The sample arm is modeled as an LTI system $h_{2}(x)$ that maps the input signal to an output
\begin{equation}\label{eq:second_arm_out_general}
\psi_{2}(x)=\sum_{n=1}^{N}d_n \phi_n(x),
\end{equation}
where $d_n, \,\,\,n=1,2,...,N$ are new basis coefficients. In temporal interferometry, the signal energies are acquired by time-averaging for each setting of the temporal delay $\tau$. Here, in generalized interferometry, the signal energy is obtained by averaging over the degree of freedom $x$ for each setting of the delay $\alpha$. Accordingly, the interferogram generated is
\begin{equation}\label{eq:interferogram_general2}
I(\alpha)=I_1+I_2+2\sum_{n=1}^{N}\!|c_n||d_n|\cos\big(n\alpha+\theta_n\big),
\end{equation}
which is analogous to (\ref{eq:interferogram_exponential}), and $I_{1}$ and $I_{2}$ represent the energy in each arm.

A number $M$ of interferometric measurements are collected by sampling the delay parameter $\alpha_{m}, \,\,\,m=1,2,...,M$. 
Therefore, similar to (\ref{eq:define_linearmeasurements}), we obtain a linear model for the interferometric measurements where,
\begin{align}\label{eq:interferogram_linear_general}
\begin{aligned}
y(\alpha_{m})\triangleq\tfrac{1}{2}(I(\alpha_m)-I_{1}-I_{2})
=\sum_{n=1}^{N}|c_n||d_n|\cos\big(n\alpha_{m}+\theta_n\big),\,m=1,2,...,M.
\end{aligned}
\end{align}
Thus (\ref{eq:interferogram_linear_general}) can also be cast in vector form as
\begin{equation}\label{eq:linear_generalizedinterf}
\mathbf{y}=\mathbf{A}\mathbf{x},
\end{equation}
where the definitions and dimensions of the information coefficient vector $\mathbf{x}$, the measurement vector $\mathbf{y}$, and the matrix $\mathbf{A}$ are identical to those in (\ref{eq:interferogram_linearmodel}) after replacing the temporal delay samples $\tau_{m}$ with the sampled generalized delay parameter $\alpha_{m}$. Similar to temporal interferometry, the measurement model in (\ref{eq:interferogram_linear_general}) enables us to retrieve information about the input signal or the sample embedded in the coefficients $c_n$ and $d_n$.

Remarkably, the result in (\ref{eq:linear_generalizedinterf}) shows that the proposed framework is in fact basis-neutral. This is clear from the fact that $\mathbf{x}$, $\mathbf{y}$, and $\mathbf{A}$ have no traces of the basis functions $\{\phi_{n}(x)\}$, which is a consequence of the diagonal representation of the generalized delay in this basis. Therefore, any analysis based on (\ref{eq:linear_generalizedinterf}) is independent of the underlying basis and applies equally to all.

To this point, we used temporal interferometry to analyze a signal or an optical field into its time-frequency harmonics. 
Instead, signals or optical fields can also be analyzed in different bases with spatial degrees of freedom, in Cartesian or polar coordinate systems, for example. Rapid and accurate modal analysis in different bases is of critical importance for terabit communications in free space \cite{Wang12NP} and multimode fibers \cite{Bozinovic13Sc} that make use of spatial multiplexing to increase the information-carrying capacity.

As a case study, we show how to leverage a Hilbert space analyzer to decompose an optical field in the Hilbert space spanned by the Hermite-Gaussian (HG) modes, which are of paramount importance in optics because they are natural modes of laser resonators \cite{siegman1986lasers}.  Consider an optical field $\psi(x)\!=\!\sum_{n=1}^{N}c_n\phi_n(x)$ consisting of a superposition of HG modes. To analyze the field into its constituent modes, the reference arm should include a unitary transformation for which the HG modes are eigenfunctions. As mentioned above, this transformation is a frFT of order $\alpha$ since the HG modes are eigenfunctions of the frFT with eigenvalues $e^{-jn\alpha}$ \cite{Ozaktas01wiley}. The kernel of an frFT system of order $\alpha$ is,
\begin{equation}\label{eq:actualfrFT}
h_{1}(x,x';\alpha)\!\propto\!\exp\left\{\frac{j\pi}{2}(x^2\cot\alpha\!+\!x'^2\cot\alpha\!-\!2xx'\csc\alpha)\right\},
\end{equation}
whose optical implementation makes use of two cylindrical lenses \cite{Abouraddy12OL,sci_rep2017}; here $x$ and $x'$ are appropriately normalized spatial coordinates. The output from the frFT is superposed with the output of the sample arm $\psi_2(t)\!=\!\psi(x)$ to acquire the interferometric measurements as in (\ref{eq:interferogram_general2}). Considering (\ref{eq:interferogram_linear_general}), the interferogram can be cast as $\mathbf{y}\!=\!\mathbf{A}\mathbf{x}$, where $\mathbf{x}^T=[|c_1|^2 |c_2|^2...|c_N|^2]$, and $\mathbf{A}\!=\![\cos(n\alpha_m)]$, $n=1,2,...,N$, and $m=1,2,...,M$.

\textcolor{black}{
\begin{remark}
\label{rem:multidim}
Our approach extends naturally to signals or fields described by multiple degrees of freedom, in which case interferometry can be performed in higher dimensions by introducing several generalized delays, one for each degree of freedom -- an example is shown in Fig. \ref{fig:combination_apps}(b). In such cases, multi-dimensional interferograms are produced by sampling the corresponding delay parameters. For brevity and to simplify the exposition, we have only presented the 1D case, however, we provide a 2D example in Section \ref{sec:multidim_example}. 
\end{remark}
}

\section{Compressive Reconstruction Under Sensing Constraints}\label{sec: CS_practical constraint}

We established that the interferograms relate to the Hilbert space coefficients via a linear operator defined by the parameters of the interferometer. It follows that basis analysis from interferometric measurements is amenable to compressive data acquisition in the sense that recovery and reconstruction can be potentially carried out using a reduced number of measurements provided $\mathbf{x}$ admits some additional structure. As mentioned earlier, this linear model holds for a wide range of problems regardless of the nature of the underlying signal domain. For example, in analyzing an optical beam into spatial modes, many of the basis coefficients corresponding to modal occupation are zero or near-zero since only a few are activated at any time in high-speed optical communications using spatial multiplexing. Clearly, $\mathbf{x}$ admits a sparse structure.

This Section is focused on compressive basis analysis from interferometric measurements under a sparsity assumption on $\mathbf{x}$ with the goal of reducing both the sampling complexity and acquisition time. This is particularly useful for scenarios where measurements are costly, as well as in delay-sensitive applications. We seek reconstruction of $\mathbf{x}$ from $M\ll N$ interferometric measurements corresponding to $M$ settings of the interferometric delay parameter $\alpha$.

We point out two fundamental differences between our approach
and prior work employing compressive techniques. First, the vast majority of prior work on compressive sensing presumes one has full control over the design of the sensing matrix -- for example, in optics, by introducing designed random masks along the path of an optical field in an imaging system \cite{Duarte08SPM, Howland14PRL,Mirhosseini14PRL}. In sharp contrast, the matrix $\mathbf{A}$ in our interferometric formulation is imposed through the structure of the interferometer itself. Therefore, compression has to be carried out under sensing constraints set by the limited degrees of freedom of the sensing system. It is not clear at the outset whether performance guarantees on reconstruction can be established given the special structure of the constrained matrix $\mathbf{A}$. Second, previous work on using compressive sensing in optical interferometry has mostly focused on reducing the number of measurements used for recovery/reconstruction, but not on compressive data acquisition. For example, in the context of OCT, the approach in \cite{baha_OCT,Liu:10} selects a random subset of many interferometric measurements collected using a CCD array detector. This amounts to using fewer measurements in the recovery of depth information and discarding measurements already collected by the physical sensing system. By contrast, our approach directly uses the degrees of freedom inherent to the sensing system (by assigning some random values to the generalized phase $\alpha$) to reduce the data acquired in the first place for subsequent recovery.

\subsection{Preliminaries}
Consider the standard $\ell_1$-minimization convex program, so-called Basis Pursuit (BP) \cite{BP}, for reconstructing $\mathbf{x}$ from $M$ measurements of the form (\ref{eq:interferogram_linearmodel}),
\begin{equation}\label{eq:BP}
\begin{split}
&\text{minimize}~ \|\mathbf{x}\|_1\\
& \text{subject to}\,\,\,\,\, \mathbf{y}=\mathbf{A}\mathbf{x}.
\end{split}
\end{equation}
It is widely recognized that (\ref{eq:BP}) can successfully recover a sparse vector $\mathbf{x}\in\mathbb{R}^{2N}$ from $M\ll 2N$ measurements provided the sensing matrix $\mathbf{A}$ satisfies some conditions \cite{Candes08CR}. For example, it has been established that an $s$-sparse vector (with at most $s$ non-zero elements) can be reconstructed using (\ref{eq:BP}) if $\mathbf{A}$ satisfies the Restricted Isometry Property (RIP), which requires that
\begin{equation}\label{eq:RIP}
(1-\delta)\|\hat{\mathbf{x}}\|_2^2\!\leq\!\|\mathbf{A}\hat{\mathbf{x}}\|_2^2\!\leq\!(1+\delta)\|\hat{\mathbf{x}}\|_2^2
\end{equation}
for any $\hat{\mathbf{x}}\in\Sigma_{2s}$, where $\Sigma_{2s}:=\{\mathbf{x}\in\mathbb{R}^{2N}: \|x\|_0\leq 2s\}$ is the set of all $2s$-sparse vectors in $\mathbb{R}^{2N}$ for a parameter $0<\delta\!<\!\sqrt{2}-1$ known as the restricted isometry constant.

\subsection{Constrained sensing}\label{sec:cons_sens}
Sub-Gaussian random sensing matrices satisfy the RIP with high probability for $M\approx\mathcal{O}(s\log N)$, which motivated their use in several CS applications. However, in practice one may not have full control over the design of the sensing matrix $\mathbf{A}$ as it is normally determined by the structure of the data acquisition system (DAQ). As such, much of the prior work in measurement and instrumentation relied on introducing additional random masks during the measurement process to emulate random sensing matrices (e.g., optical imaging \cite{Duarte08SPM} and field estimation \cite{Mirhosseini14PRL}). This requires modifying the actual DAQs thereby incurs additional cost and complexity.

In sharp contrast, we show here that CS can be exploited in `native' interferometry, that is, without modifying the underlying interferometer structure nor introducing additional components. Recalling that the rows of the sensing matrix $\bA$ found in (\ref{eq:linear_generalizedinterf}) have the $\alpha$-dependent structure
\begin{equation}
\mathbf{a}_m=[\ba_{m1} \,\,\, \mathbf{a}_{m2}], \,\,\,m=1,2,...,M,
\end{equation}
where,
\begin{equation}
\begin{split}
&\mathbf{a}_{m1}^{\mathrm{T}}=[\cos(n\alpha_m)],~
\mathbf{a}_{m2}^{\mathrm{T}}=[-\sin(n\alpha_m)],
\end{split}
\end{equation}
we see that $\bA$ has only few degrees of freedom corresponding to the settings of the generalized delay parameter $\alpha$. The rest of this section focuses  on signal reconstruction based on compressive interferometric measurements of the form (\ref{eq:linear_generalizedinterf}) and establishing performance guarantees thereof.

\subsection{Guarantees with randomized delays}\label{ssec:random}
Collecting informative interferometric measurements (\ref{eq:linear_generalizedinterf}), and in turn achieving better performance in reconstruction, is premised on selecting appropriate values for the generalized delay parameter $\alpha$. We consider sensing matrices generated by drawing generalized delays from random distributions.

Throughout this section, we consider normalized interferometric measurements $\by = \hat{\bA}\bx$, where $\hat{\bA} = \sqrt{2/M}\bA$ and $\bA$ the original matrix defined in (\ref{eq:linear_generalizedinterf}). Our next theorem establishes that the matrix $\hat{\bA}$ is RIP provided the generalized delay parameters are selected from an appropriate distribution. 

\begin{theorem}\label{th:RIP_nonasymptotic_general}
If the generalized delay parameters $\alpha_m,\,m=1,2,\ldots,M$, of the matrix $\mathbf{A}$ in (\ref{eq:linear_generalizedinterf}) are chosen independently and uniformly at random from distribution ${\cal U}[0,2\pi]$, then there exist positive constants $c_1$, $c_2$ such that $\hat{\mathbf{A}}:=\sqrt{2/M}\bA$ satisfies the RIP in (\ref{eq:RIP}) with respect to all $s$-sparse vectors with any $s\!\leq\!c_1M/\log(2N/s)$, and an RIP constant $0\!<\delta\!<1$ with probability greater than $1-2e^{-c_2M}$,  where $c_2 \!\leq\!c_0(\delta/2)-c_1[1+(1+\log(12/\delta))/\log(2N/s)]$.
\end{theorem}

\begin{proof}\label{proof:RIP_nonasymptotic_general}
Following the procedure in \cite{Simpleproof}, it suffices to show that $\hat\bA$ satisfies the concentration inequality
\begin{equation}\label{eq:concentration}
\mathbb{P}\{|\|\hat{\mathbf{A}}\mathbf{x}\|^2-\|\mathbf{x}\|^2|\!\geq\!\epsilon \|\mathbf{x}\|^2\}\!\leq\!2e^{-Mc_0(\epsilon)}, \,\,\,\,0<\epsilon<1
\end{equation}
for all $\mathbf{x}\in\Sigma_s$ under the condition in the statement of Theorem \ref{th:RIP_nonasymptotic_general}.
Since the $M$ realizations $\alpha_m, m = 1,\ldots, M$, are selected independently from a random uniform distribution ${\cal U}[0,2\pi]$, $\|\hat{\mathbf{A}}\mathbf{x}\|^2$ can be written as a sum of $M$ i.i.d. random variables,
\begin{equation}
\|\hat{\mathbf{A}}\mathbf{x}\|^2=\sum_{m=1}^{M}|\langle \hat{\mathbf{a}}_m,\mathbf{x}\rangle|^2,\,\,\,\,\forall \mathbf{x}\!\in\!\Sigma_s,
\end{equation}
where $\langle.,.\rangle$ denotes the inner product of its two vectors argument. Assuming a fixed but arbitrary vector $\mathbf{x}_0\in\Sigma_s$, each random variable $Z_m\!\triangleq\!|\!<\!\hat{\mathbf{a}}_m,\mathbf{x}_0\!>\!|^2$ can be bounded as,
\begin{equation}
\begin{split}
Z_m\!\triangleq|\!<\!\hat{\mathbf{a}}_m,\mathbf{x}_0\!>\!|^2\!\leq\!\|\hat{\mathbf{a}}_m\|^2\cdot\|\mathbf{x}_0\|^2\!\leq\!\frac{2s}{M}\|\mathbf{x}_0\|^2, \,\,\,m=1,2,\ldots,M,
\end{split}
\end{equation}
using the Cauchy-Schwarz inequality \cite{Cauchy}.
Hence, the random variable $\|\hat{\mathbf{A}}\mathbf{x}_0\|^2$ is a summation of $M$ bounded random variables $Z_m\!\in\![0,\frac{2s}{M}\|\mathbf{x}_0\|^2]$. Accordingly, using Hoeffding's inequality \cite{Hoeffding} we have
\begin{equation}\label{eq:concentration_before_final}
\begin{split}
\mathbb{P}\{|\|\hat{\mathbf{A}}\mathbf{x}_0\|^2-\mathbb{E}{\|\hat{\mathbf{A}}\mathbf{x}_0\|^2}|\!\geq\!\epsilon \|\mathbf{x}_0\|^2\}
\leq\!2e^{-\frac{2\epsilon^2\|\mathbf{x}_0\|^4}{\frac{4s^2}{M}\|\mathbf{x}_0\|^4}}
=2e^{-M\frac{\epsilon^2}{2s^2}}, \,\,\,\,0<\epsilon<1.
\end{split}
\end{equation}
Given the distribution of $\alpha$, $\mathbb{E}{\|\hat{\mathbf{A}}\mathbf{x}\|^2}=\|\mathbf{x}\|^2$ for all $\mathbf{x}\!\in\!\Sigma_s$. Thus, we can rewrite (\ref{eq:concentration_before_final}) as,
\begin{equation}\label{eq:concentration_final}
\begin{split}
\mathbb{P}\{|\|\hat{\mathbf{A}}\mathbf{x}\|^2-{\|\mathbf{x}\|^2}|\!\geq\!\epsilon \|\mathbf{x}\|^2\}\!\leq\!2e^{-Mc_0(\epsilon)},
\,\,\,\, 0<\epsilon<1, \,\,\forall \,\, \mathbf{x}\!\in\!\Sigma_s,
\end{split}
\end{equation}
where $c_0(\epsilon)=\epsilon^2/2s^2$. Hence, it follows from \cite[Theorem 5.2]{Simpleproof} that the matrix $\hat{\mathbf{A}}$ is RIP with respect to all $\mathbf{x}\!\in\!\Sigma_s$ with RIP constant $0\!<\delta<\!1$, with probability greater than $1-2e^{-c_2M}$,  where $c_2 \!\leq\!c_0(\delta/2)-c_1[1+(1+\log(12/\delta))/\log(2N/s)]$.
\end{proof}

Based on Theorem \ref{th:RIP_nonasymptotic_general}, the sensing matrix $\hat{\mathbf{A}}$ satisfies the RIP with higher probability as the number of measurements $M$ increases. The next corollary identifies an asymptotic regime where the sensing matrix satisfies the RIP with probability $1$. 

\begin{corollary}\label{cor:asymptotic}
The sensing matrix $\hat{\mathbf{A}}$ defined in Theorem \ref{th:RIP_nonasymptotic_general} satisfies the RIP with a constant $0\!<\delta\!<1$~ for all $s$-sparse vectors with probability $1$, if $N\!\rightarrow\!\infty$ and $M\!\rightarrow\!\infty$. 
\end{corollary}
\begin{proof}\label{proofcor:asymptotic}
The proof follows directly from the fact that $c_2$ is always a positive constant, so  
\begin{equation}
\lim_{M\!\rightarrow\!\infty} -c_2 M=-\infty
\end{equation}

given the asymptotic order of $M$ in the statement of the corollary. Therefore, the probability $1-2e^{-c_2M}\!\rightarrow\!1$.
\end{proof}

The results of Theorem \ref{th:RIP_nonasymptotic_general} and Corollary \ref{cor:asymptotic} are general in that they apply to every problem in interferometry with the measurement model in (\ref{eq:linear_generalizedinterf}). We have already established the generality of the framework that gave rise to (\ref{eq:linear_generalizedinterf}), which was also shown to be basis-neutral. As a direct application of this result, the following corollary establishes that the matrix arising in optical modal analysis at the end of Section \ref{sec:generalmodel} is also RIP.
\begin{corollary}\label{cor:optics}
Given $0\!<\delta<\!1$ and $s\!\leq\!c_1M/\log(N/s)$, the sensing matrix $\hat{\mathbf{A}}=[\sqrt{2/M}\cos(n\alpha_m)]$ arising in the (generalized) optical modal analysis example (which consists of only the cosine terms), is RIP with respect to all $s$-sparse vectors in $\mathbb{R}^{N}$ with probability greater than $1-2e^{-c_2M}$, where where $c_2 \!\leq\!c_0(\delta/2)-c_1[1+(1+\log(12/\delta))/\log(N/s)]$.
\end{corollary}
The proof follows directly from Theorem \ref{th:RIP_nonasymptotic_general} and by replacing $2N$ with $N$ in the probability bound.

Theorem  \ref{th:RIP_nonasymptotic_general} established a lower bound on the probability that $\hat{\mathbf{A}}$ is RIP, which goes asymptotically to $1$ per Corollary \ref{cor:asymptotic}. In non-asymptotic regimes and when the number of measurements is not sufficiently large, this bound can be fairly far from $1$. It turns out that the constrained matrix ${\mathbf{A}}$ also satisfies some weaker sufficient conditions for recoverability when the generalized delay parameters are drawn uniformly at random. In particular, we establish that the ensemble of sensing matrices corresponding to $\alpha$'s drawn from a uniform distribution ${\cal U}[0,2\pi]$ is isotropic and incoherent \cite{RIPless}, therefore an arbitrary fixed sparse vector $\mathbf{x}$ can be reconstructed from compressive measurements with high probability \cite{RIPless}. First, we review the definitions of the isotropy and incoherence properties, then state our result.
\begin{definition}(Isotropy \cite{RIPless}) If the vector $\mathbf{g}$ denotes a row of a random matrix $\bG$ drawn from a probability distribution $F$, then $F$ is said to satisfy the isotropy property if
\begin{equation}\label{eq:isotropy}
	\mathbb{E}[\mathbf{g}^H\mathbf{g}]=\mathbf{I},
	\end{equation}
where $\mathbb{E}[.]$ denotes the expectation and $\mathbf{I}$ the identity matrix.
\end{definition}
\begin{definition}
(Incoherence \cite{RIPless}) The distribution $F$ of $\bg = [g_n]\in\mathbb{C}^N$, is said to be incoherent with incoherence parameter $\mu(F)$ if
	\begin{equation}\label{eq:incoherence}
	\underset{n=1,2,...,N}{\mbox{max}}|g_n|^2\leq \mu,
	\end{equation}
where $\mu$ is the smallest number for which (\ref{eq:incoherence}) is satisfied. 
\end{definition}
The smaller the incoherence parameter $\mu$, the less the number of measurements $M$ required for (\ref{eq:BP}) to yield successful reconstruction \cite[Theorem 1.1]{RIPless}. It was also shown in that denoising algorithms such as LASSO \cite{LASSO} and the Dantzig \cite{Dantzig} selector yield stable recovery from noisy measurements under the isotropy and incoherence properties of the sensing matrix (Theorems 1.2 and 1.3 in \cite{RIPless}).

Henceforth, we refer to the matrix $\mathbf{A}$ as \emph{isotropic and incoherent} if the distribution $F$ of its rows (specified by the generalized delay parameter) obeys the isotropy and incoherence properties in (\ref{eq:isotropy}) and (\ref{eq:incoherence}). We can readily state the following lemma which establishes sufficient conditions for successful reconstruction from interferometric measurements based on the generalized interferometry framework. 

\begin{lemma}\label{The: CS ideal performance}
	Suppose $M$ interferometric measurements are acquired by selecting the generalized delay parameters $\alpha_m, m=1,2,\ldots,M$, of $\bA$ in (\ref{eq:linear_generalizedinterf})  from a uniform distribution ${\cal U}[0,2\pi]$. If $M\geq 2 L_0(1+\beta) s\log (2 N)$ for a positive constant $L_0$ and any $\beta > 0$, the $\ell_1$-norm minimization in (\ref{eq:BP}) yields the $s$-sparse vector $\mathbf{x}\in\mathbb{R}^N$ from the normalized measurements $\mathbf{y}=\hat{\mathbf{A}}\mathbf{x}$ with probability at least $1-\frac{5}{2N}-e^{-\beta}$.
\end{lemma}
\begin{proof}
{\color{black} Based on \cite[Theorem 1.1]{RIPless}, we only need to show that $\sqrt{M}\, \hat{\mathbf{A}}$ is incoherent and isotropic under the conditions in the statement of Lemma \ref{The: CS ideal performance}.

It is easy to see that for the matrix $\hat{\mathbf{A}}$, $\max_{n=1,2,\ldots,2N}|\hat{a}_{m,n}|^2\leq2/M$, as the cosine and sine terms are bounded below and above by $-1$ and $1$, respectively. So, the matrix $\sqrt{M}\,\hat{\mathbf{A}}$ is incoherent with parameter $\mu=2$. Also, if $\alpha_m\sim{\cal U}[0,2\pi]$, then $\mathbb{E}[\hat{\mathbf{a}}_m^H\hat{\mathbf{a}}_m]=(1/M)\,\mathbf{I}$, therefore $\sqrt{M}\,\hat{\mathbf{A}}$ is isotropic.} Accordingly, Lemma \ref{The: CS ideal performance} follows from \cite[Theorem 1.1]{RIPless}.
\end{proof}

We also consider the noisy case,
\begin{equation}\label{eq:csnoise}
\mathbf{y}=\hat{\mathbf{A}}\mathbf{x}+\mathbf{z},
\end{equation}
where $\mathbf{z}\sim\mathcal{N}(0,\sigma^2I)$. A sufficient condition on the number of measurements for successful reconstruction in presence of noise is stated next.

\begin{lemma}\label{The: CS noisy performance}
	
Consider the same setting in the statement of Lemma \ref{The: CS ideal performance}. For any $\beta > 0$,  if the number of noisy measurements $M\geq L_0.(1+\beta).2.s\log(2N)$, then the
LASSO algorithm \cite{LASSO} with parameter $\lambda_{\text{LASSO}}=10\sqrt{\log(2N)}$ yields a vector $\mathbf{\bar{x}}$ satisfying
	\begin{equation}\label{eq:noisyCS_bound}
		\|\mathbf{\bar{x}-x}\|_2\leq \min_{1\leq\bar{s}\leq s}\zeta(\bar{s}),
	\end{equation}
	where,
	\begin{equation}\label{eq:perfc_bound}
	\zeta(\bar{s})\!\triangleq\!L(1+\gamma)\Big(\frac{\|\mathbf{x-x}_{\bar{s}}\|_1}{\sqrt{\bar{s}}}+\sigma\sqrt{\bar{s}\log(2N)}\Big)\\
	\end{equation}
	with probability at least $1-6/(2N)-6e^{-\beta}$, where $L$ is a positive constant and $\gamma=\sqrt{\frac{(1+\beta)2\bar{s}\log(2N)\log{M}\log^2{\bar{s}}}{M}}$.
\end{lemma}
\begin{proof}
	Similar to the proof of Lemma \ref{The: CS ideal performance}, Lemma \ref{The: CS noisy performance} follows from the incoherence and isotropy of $\sqrt{M}\,\hat{\mathbf{A}}$ and the results of Theorem 1.2 and 1.3 in \cite{RIPless}.
\end{proof}
Instead of LASSO, we can use the Dantzig selector \cite{Dantzig} to recover the spare vector in noise. In this case, under the same conditions of Lemma \ref{The: CS noisy performance}, the performance bound provided for the reconstruction error in (\ref{eq:perfc_bound}) is still valid by replacing $\gamma$ with $\gamma^2$ \cite{candesRIPless}.

\section{simulation and experimental results}\label{sec:results}
{\color{black} In this section, we study three different examples to evaluate the performance of the proposed approach to signal reconstruction from interferometric measurements under sensing constraints. In the first example, we consider optical modal analysis for Hermite Gaussian beams. This is an example of the generalized interferometry framework introduced in Section \ref{sec:generalmodel} in the spatial degree of freedom, and of compressive reconstruction under sensing constraints. In the second example, we reconstruct a layered sample placed in one arm of the interferometer using TD-OCT. This complements the former modal analysis example in which the information of interest pertains to the input signal, yet both examples are studied within the same unifying basis analysis framework. Lastly, we present an example of multi-dimensional interferometry using two-dimensional HG signals with two spatial degrees of freedom as per Remark \ref{rem:multidim}.} 

\subsection{Optical modal analysis}\label{ssec:modal_exp}
As discussed earlier, the HG modes are the eigenfunctions of an frFT of order $\alpha$ with eigenvalues $e^{-jn\alpha},\,\,n=1,2,...,N$. To analyze an optical beam in a Hilbert space spanned by the HG modes, we collect $M$ interferometric measurements by selecting $M$ different frFT orders $\alpha_m,\,\,m=1,2,...,M$, then apply a CS reconstruction method to reveal the modal content of the beam (the basis coefficients). In this case, $\by = \bA\bx$, where the $M\!\times\!N$ matrix  $\mathbf{A}=[\cos(n\alpha_m)]$ and $\mathbf{x}^{\mathrm{T}}=[|c_1|^2 |c_2|^2...|c_N|^2]$ is s-sparse. The frFT orders $\alpha_m,\,\,\,m=1,2,...,M$, specifying the rows of $\bA$ are i.i.d. and drawn from a uniform distribution $\mathcal{U}[0,2\pi]$, thus $\bA$ is isotropic and incoherent. The possible number of modes is set to $N=64$, and the number of active modes is $s=4$. In presence of noise, $\by = \bA\bx+\bz$, where the noise $\mathbf{z}\sim\mathcal{N}(0,\sigma^2\mathbf{I})$ is white and Gaussian, and $\mbox{SNR}\triangleq10\log(\frac{\bx^\text{H} \mathbb{E}[\bA^{\text{H}}\bA]\bx}{\sigma^2})$, where $\mathbb{E}[.]$ denotes the expectation over the distribution of $\alpha_m$ parametrizing $\bA$. To evaluate the quality of reconstruction, the scaled recovery error is defined as $e\triangleq\frac{\|\mathbf{x}-\bar{\mathbf{x}}\|_2^2}{\|\mathbf{x}\|_2^2}$. We use the BP and the Dantzig selector algorithms to reconstruct the modal coefficients in noise-free and noisy environments, respectively.

\subsubsection{Results with ideal frFTs}\label{ssec:ideal}
We first assume an ideal implementation for the frFT of different orders. Accordingly, for order $\alpha_m$, the output beam is $\psi_1(x;\alpha_m)=\sum_{n=1}^{N}c_n\phi_n(x)e^{-jn\alpha_m}$ as in (\ref{eq:filter output}), where $\alpha_m, \,\,m=1,2,...,M$ are selected independently and identically from the uniform distribution $\mathcal{U}[0,2\pi]$. 

Based on the generalized framework introduced in Section \ref{sec:generalmodel} which works for arbitrary basis, an FT of the interferometric measurements can be used to recover the modal energies of the input beam. Since the largest mode order is $N$, sampling uniformly at the Nyquist rate amounts to collecting $2N=128$ measurements by selecting the orders of the frFT uniformly and deterministically between $0$ and $2\pi$. In this case,
\begin{equation}
\bar{\mathbf{x}}=|\mathbf{F}\mathbf{y}|,
\end{equation}
where $\mathbf{F}$ is a $2N\!\times\!2N$ DFT matrix. While in the FT approach $M=2N=128$ interferometric measurements are needed for successful recovery, Fig. \ref{fig:error_ideal} shows that the modal content can be retrieved with significantly less measurements if we use the CS approach. Despite the constrained structure of $\bA$, from only $M=25$ measurements the CS approach yields reconstruction performance comparable to that of FT whilst achieving substantial savings in data acquisition time. 

\begin{figure*}[htb]
	\centering
	\includegraphics[scale=.9]{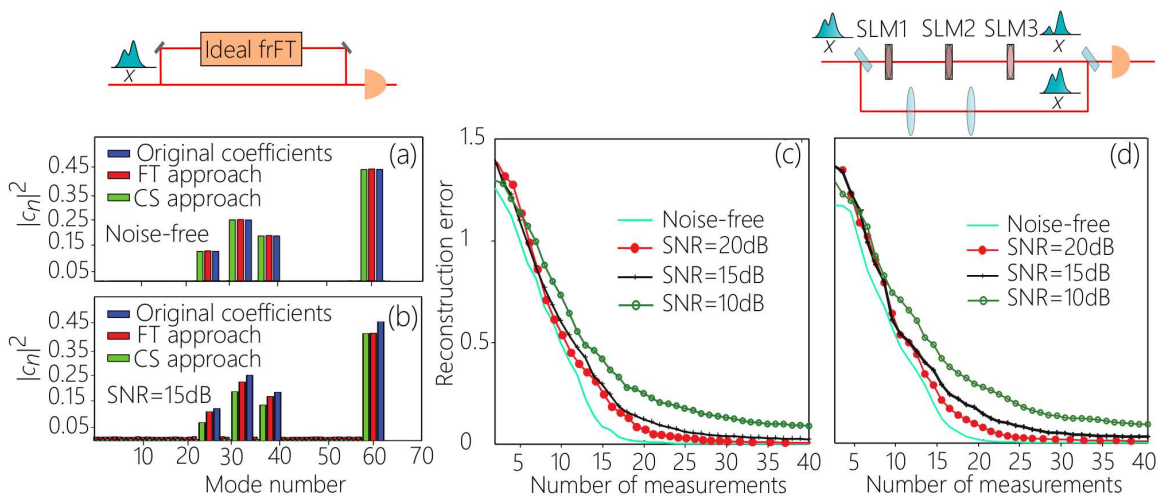}
	\vspace{-.2cm}
	\vspace{-3pt}\caption{(a) Comparing the CS approach performance with $M=25$ measurements collected using an ideal frFT to that of the FT approach with $2N=128$ measurements in noise-free environment. (b)  SNR$=15$dB, and ideal frFT. (c) Evaluating the performance of the CS approach in terms of the reconstruction error as function of the number of interferometric measurements for different SNRs. The frFT filters are assumed to be ideal, and $N=64$, $s=4$. (d) Evaluating the performance of the CS approach in terms of reconstruction error versus the number of interferometric measurements for different SNRs. The frFT filters are simulated using ideal SLMs acting as quadratic phase operators, and $N=64$, $s=4$.}
	\label{fig:error_ideal}
\end{figure*}

We investigate the recovery/reconstruction performance by calculating the reconstruction error for a different number of measurements in both noise-free and noisy settings. Fig. \ref{fig:error_ideal}(c) shows the decay of the reconstruction error with $M$ using the compressive approach. The curves are obtained by averaging the error over $100$ runs. The sparse vector is held fixed across the different runs, but the rows of the sensing matrix are randomly generated from the uniform distribution giving a different interferogram for every run.

\begin{figure}[htb]
	\centering
	\includegraphics[scale=.9]{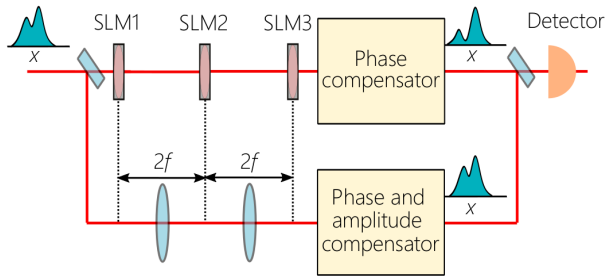}
	\vspace{-.4cm}
	\caption{\textcolor{black}{Schematic of an frFT filter implemented using SLMs that act as quadratic phase operators.}}
	\label{fig:frFT_schematic}\vspace{-10pt}
\end{figure}
\subsubsection{Results with the actual implementation of the frFT}\label{ssec:simfrFT}
Figure \ref{fig:frFT_schematic} shows an interferometer with a practical optical setup implementing an frFT filter. The setup consists of three Spatial Light Modulators (SLMs), which suitably modulate the phase of the optical field along the transverse direction. The SLMs act as quadratic phase operators -- and this is how they are modeled in our simulations -- depending on $\alpha_m$, as $e^{\frac{-j\pi p_{\ell}(\alpha_m)x^2}{2\lambda}}$, where $\ell=1,2,3$, is the index of the SLM, and the phase function $p_{\ell}(\alpha_m)$ is defined by the generalized delay parameter and the characteristics of the input beam. \textcolor{black}{In appendix}, we show how this system actually implements an frFT using properly chosen $p_{\ell}(\alpha_m), \ell=1,2,3$. 

Figure \ref{fig:error_ideal}(d) displays the simulated reconstruction error versus the number of measurements $M$ for the CS approach using the frFT system shown in the schematic of Fig. \ref{fig:frFT_schematic}. {\color{black}The performance is fairly close to the ideal case demonstrating that the three-SLM setup accurately models an frFT.}

\subsubsection{Experimental results}\label{experimental}
Here, we report on results from an actual laboratory experiment implementing the frFT filter of Fig. \ref{fig:frFT_schematic}.  Producing exact HG modes is practically infeasible. Instead, we obtain approximate modes by flipping a Gaussian beam at the crossing points (see the approximate modes of the incident beam in the insets of Fig. \ref{fig:experiment_coef}). Obviously, such beams are not perfectly orthogonal, hence will have non-vanishing mutual projections. As such, even if a single mode is active, there will be non-zero coefficients for the adjacent modes. 

In Fig. \ref{fig:experiment_coef}, we compare the performance of the CS approach to that of FT for an optical beam consisting of HG$_1$ and another of HG$_2$, where HG$_1$ and HG$_2$ are the first and second Hermite-Gaussian modes, respectively. The FT approach uses $128$ interferometric measurements collected uniformly by choosing the generalized delays between $0$ to $2\pi$. In the CS approach, only $M=25$ random measurements are used.

\begin{figure}[htb]
	\centering
	\includegraphics[scale=.85]{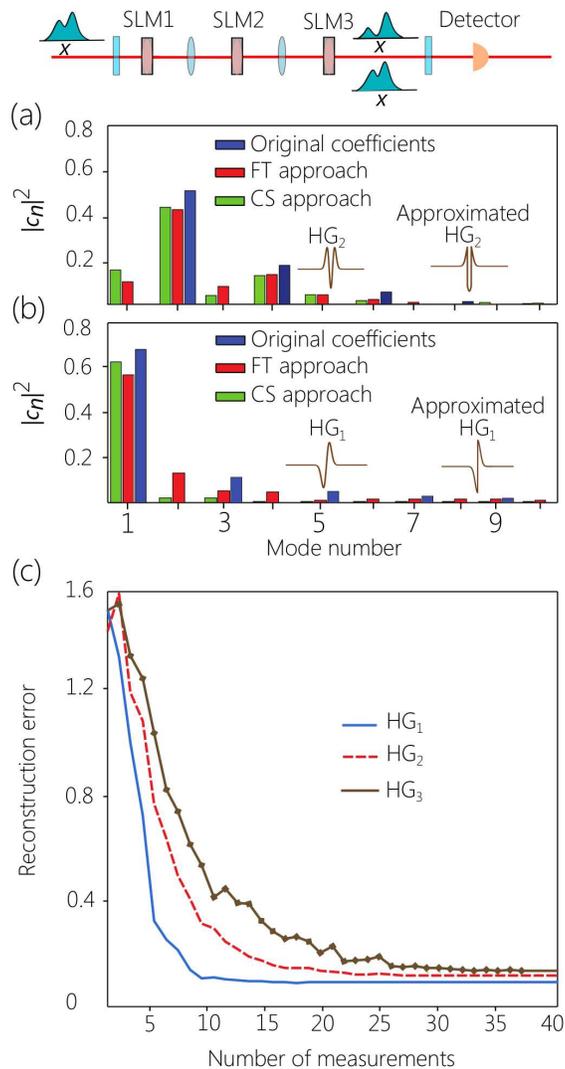}
\vspace{-10pt}	\caption{Comparing the reconstruction performance of the CS approach to that of the FT from experimental measurements. (a) Using approximate HG$_1$ mode. (b) Using approximate HG$_2$ mode. (c) Evaluating the performance of the CS approach in the experiment using approximate HG$_1$, HG$_2$, and HG$_3$ in terms of reconstruction error versus the number of interferometric measurements for different values of SNR. The frFT is implemented using real optical devices with various physical constraints.}
	\label{fig:experiment_coef}
\end{figure}

We also investigate the reconstruction error based on the experimental results. As shown in Fig. \ref{fig:experiment_coef}, efficient reconstruction requires about $M=25$ measurements. This corresponds to $25$ settings of the frFT order for the CS approach versus $128$ for FT. 

{\color{black}
\subsection{Information recovery in TD-OCT}\label{ssec:OCT_exp}
Here, we present an example of TD-OCT in which we seek to recover the reflectivity and depth information of $L$ different layers of a sample object within our unifying basis analysis framework. The depth information of a given layer with respect to the sample surface is characterized in terms of the round-trip time of the optical field from the layer to the surface of the sample. As described in Section \ref{ssec:applications}, the desired information here is in the basis coefficients $d_n,\,\,\,n=1,2,...,N$, which can be retrieved by solving the system of linear equations in (\ref{eq:interferogram_linearmodel}). 

In this experiment, we first consider a sample object with $L=10$ layers. By solving  (\ref{eq:interferogram_linearmodel}), we reconstruct the $20\!\times\!1$ vector $\mathbf{x}$ depicted in Fig. \ref{fig:OCT_nocompression}(a), which is shown to match the ground truth. Subsequently, the reflectivity of the layers and their depths are correctly reconstructed from the retrieved coefficients $d_n$ as displayed in Fig. \ref{fig:OCT_nocompression}(b).}
\begin{figure}[b]
	\centering
	\includegraphics[scale=1]{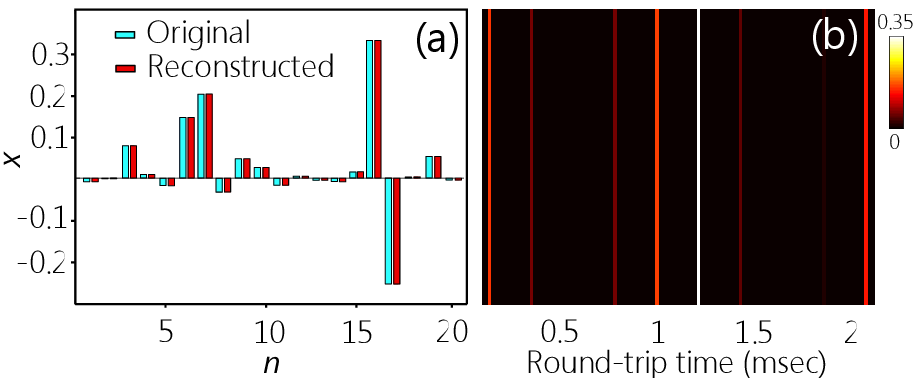}
\vspace{-.65cm}	\caption{{\color{black} Reconstructing the depth information of a sample object in TD-OCT using the basis analysis framework. (a) Entries of the reconstructed vector $\mathbf{x}$ in (\ref{eq:interferogram_linearmodel}). (b) Reconstructed reflectivity and depth information for the layers of the sample object. The depth information is characterized in terms of the round-trip time from the layer to the sample surface.}}
\label{fig:OCT_nocompression}
\end{figure}

{\color{black} We consider a second example of OCT where the sample object has $L=100$ layers among which only $s=5$ unknown layers have non-zero reflectivity. The sparsity of the vector of reflectivity indices enables recovery from few measurements. The reflectivity coefficients are successfully retrieved using Basis Pursuit and the Dantzig selector from $M=60$ interferometric measurements as shown in Fig. \ref{fig:OCT_compression}(a) and (b) for noise-free and noisy settings ($\mathrm{SNR} = 20\:\mathrm{dB}$).
\begin{figure}[htb]
	\centering
	\includegraphics[scale=1]{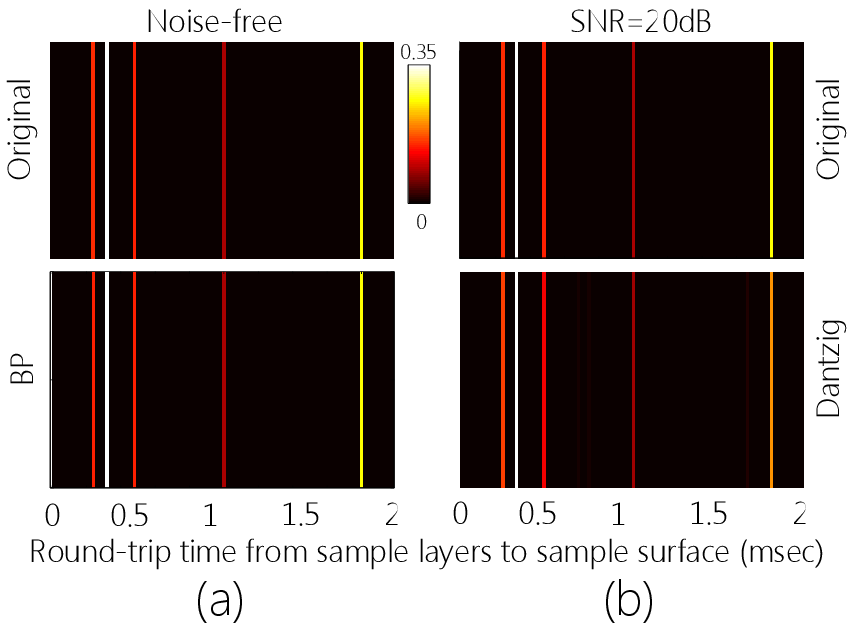}
\vspace{-10pt}	\caption{ {\color{black} Reconstructing the reflectivity indices of a sample object with $L=100$ layers and $s=5$ reflective layers from compressed interferometric measurements using BP and Dantzig selector. (a) Noise-free case. (b) Noisy setting with $\mathrm{SNR}= 20\:\mathrm{dB}$.}}
\label{fig:OCT_compression}
\end{figure}

{\color{black}
\subsection{Multi-dimensional interferometry}
\label{sec:multidim_example}
In our last example, we reconstruct a two-dimensional signal described by two spatial degrees of freedom from interferometric measurements collected by sampling two generalized delay parameters. The signal $E(x,y)$ with degrees of freedom $x$ and $y$ lies in a Hilbert space spanned by a 2D HG basis. The signal can be expressed as $E(x,y)=\sum_{nm}c_{nm} \phi_n(x)\eta_m(y)$, where $\{\phi_n(x)\}$ and $\{\eta_m(y)\}$ are the two sets of HG basis elements and $c_{nm}$ the expansion coefficients. We implement two generalized delays, namely two cascaded frFT systems of orders $\alpha_1$ and $\alpha_2$. For $N=100$ basis elements, We examine the performance of our approach in reconstructing signals formed by the superposition of a small number $s$ of basis elements. Rows (a) and (b) of Fig. \ref{fig:multi_exp} display the 2D signals and the reconstructed coefficient(s) for $s=1$ and $s=4$, respectively. Our approach is shown to yield accurate reconstruction of the expansion coefficients from a small number of interferometric measurements $M=50$, a saving of $75\%$ in sample complexity compared to directly taking a FT of the resulting interferogram. Most notably, this example underscores the ability of the proposed approach to handle spatially-multiplexed signals commonly used, for example, in high-speed communications.}

\begin{figure}[htb]
	\centering
	\includegraphics[scale=1]{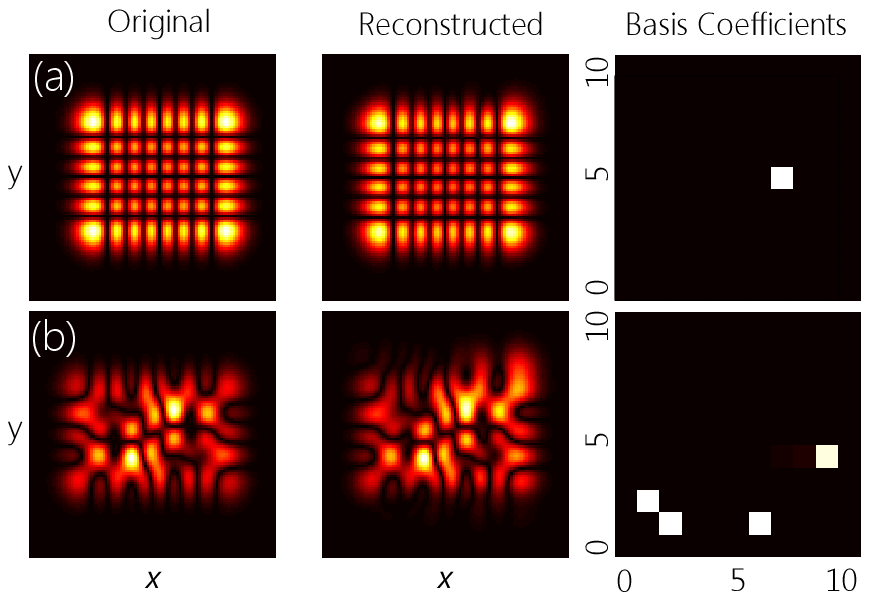}
\vspace{-10pt}	\caption{{\color{black} Reconstruction of 2D HG signals (described by two spatial degrees of freedom) from compressive interferometric measurements. The measurements are collected by sweeping the orders $\alpha_1$ and $\alpha_2$ of two cascaded frFTs for which the 2D HG signals are eigenfunctions. (a) (left) Original signal HG$_{57}=\phi_5(x)\eta_7(y)$ i.e, $(m,n) = (5,7)$, (middle) reconstructed signal and (right) reconstructed coefficient $c_{57}$. (b) (left) Original 2D signal formed by the superposition of the, $s=4$, 2D HG basis elements $(1,2), (2,1), (6,1), (9,4)$, (middle) reconstructed signal, and (right) reconstructed coefficients.}}
\label{fig:multi_exp}
\end{figure}

\textcolor{black}{
\section{Conclusion}\label{sec:conclusion}
We proposed a unifying framework for signal reconstruction from interferometric measurements. In the proposed framework, the interferometry problem amounts to Hilbert space analysis in which the information of interest resides in the expansion coefficients in a linear measurement model. This framework was generalized to arbitrary degrees of freedom enabling signal recovery in various domains beyond the time domain. Underpinning this generalization is the notion of generalized delays, namely unitary transformations for which the basis elements are eigenfunctions. The proposed framework enabled compressive acquisition and reconstruction, most notably with no modification to the underlying interferometer structure. We established performance guarantees for successful signal recovery from reduced interferometric measurements even with the constrained sensing structure set by the interferometer obviating the need for non-native components. This was validated using both synthetic and experimental data in the context of optical modal analysis, as well as examples of TD-OCT and multi-dimensional interferometry. Studying the effect of other physical sensing constraints arising from hardware limitations on signal reconstruction (e.g., limited optical aperture, clipping effects, pixelation of SLMs, etc), developing recovery algorithms under such constraints, and establishing performance guarantees thereof are subject of our ongoing investigations and will be published in subsequent work.}  

\textcolor{black}{\appendix
\noindent\textbf{Implementation of the frFT:} We show that the system in the reference arm of the interferometer depicted in Fig. \ref{fig:frFT_schematic} provides an implementation of the fractional Fourier transform (frFT). For this purpose, we model the free space propagation between the SLMs for distance $2f$, where $f$ is the focal length, with a Fresnel (diffraction) integral with kernel,
\begin{equation}\label{Fresnel_kernel}
k_{\text{Fresnel}}(x,x')=\frac{e^{j4\pi f/\lambda}}{\sqrt{j\lambda 2f}}e^{\frac{j\pi}{\lambda 2f}(x'-x)^2},
\end{equation}
where $\lambda$ is the wave length. \textcolor{black}{Higher order HG modes are obtained by multiplying the Gaussian beam $\phi_0(x)=\exp(\frac{-\pi x^2}{2\sigma^2})$ with scaled Hermite polynomials, $H_n(\frac{\sqrt{2\pi}x}{\sigma})$, where $x$ is the spatial variable in the transverse plane, $H_0(x)=1$, $H_1(x)=2x$, $H_2(x)=4x^2-2$, and $H_3(x)=8x^3-12x$, etc. \cite{Ozaktas01wiley}.}
\textcolor{black}{The first and third SLMs are modeled with multiplicative kernel $e^{-j\frac{\pi p_1 x^2}{2\lambda}}$, and the second SLM with kernel $e^{-j\frac{\pi p_2 x^2}{2\lambda}}$, where $p_1(\alpha)=\frac{-1}{f}(\csc(\alpha)+\cot(\alpha)-1)$, $p_2(\alpha)=\frac{-1}{f}(\sin(\alpha)-2)$, and $\alpha$ is the generalized delay parameter. It follows that the kernel of the overall system is
\begin{equation}\label{eq:wholesys_kernel}
\begin{split}
k_{\text{eq}}(x,u)=\!\!\frac{e^{j\zeta}}{\sqrt{2\lambda f}} \frac{1}{\sqrt{j\sin(\alpha)}}\! \cdot e^{\frac{j\pi}{2\lambda f}(x^2\cot(\alpha)+u^2\cot(\alpha)-2xu\csc(\alpha))}
	,
\end{split}
\end{equation}
where $\zeta=\frac{8\pi f}{\lambda}$.
Comparing the kernel in (\ref{eq:wholesys_kernel}) to the frFT kernel,
\begin{equation}
\begin{split}
k_{\text{frFT}}(x,u)=\frac{{\sqrt{1-j\cot(\alpha)}}}{\sqrt{2\lambda f}} e^{\frac{j\pi}{2\lambda f}(x^2\cot(\alpha)+u^2\cot(\alpha)-2xu\csc(\alpha))},
\end{split}
\end{equation}
we see that the system in the reference arm of Fig. \ref{fig:frFT_schematic} followed by an $\alpha$-dependent phase compensator block with phase $-\zeta+ \frac{\alpha}{2}$ provides an optical implementation of an frFT of order $\alpha$.}}

\bibliographystyle{myIEEEtran}
\bibliography{IEEEabrv,references_TSP1}

\begin{thebibliography}{10}
\providecommand{\url}[1]{#1}
\csname url@samestyle\endcsname
\providecommand{\newblock}{\relax}
\providecommand{\bibinfo}[2]{#2}
\providecommand{\BIBentrySTDinterwordspacing}{\spaceskip=0pt\relax}
\providecommand{\BIBentryALTinterwordstretchfactor}{4}
\providecommand{\BIBentryALTinterwordspacing}{\spaceskip=\fontdimen2\font plus
\BIBentryALTinterwordstretchfactor\fontdimen3\font minus
  \fontdimen4\font\relax}
\providecommand{\BIBforeignlanguage}[2]{{%
\expandafter\ifx\csname l@#1\endcsname\relax
\typeout{** WARNING: IEEEtran.bst: No hyphenation pattern has been}%
\typeout{** loaded for the language `#1'. Using the pattern for}%
\typeout{** the default language instead.}%
\else
\language=\csname l@#1\endcsname
\fi
#2}}
\providecommand{\BIBdecl}{\relax}
\BIBdecl

\bibitem{Pan2017_radio}
H.~Pan, T.~Blu, and M.~Vetterli, ``Towards generalized {FRI} sampling with an
  application to source resolution in radioastronomy,'' \emph{IEEE Trans. on
  Signal Processing}, vol.~65, no.~4, pp. 821--835, Feb 2017.

\bibitem{Thompson2007_radiobook}
A.~R. Thompson, J.~M. Moran, and G.~W. Swenson, \emph{Interferometry and
  Synthesis in Radio Astronomy}.\hskip 1em plus 0.5em minus 0.4em\relax
  Wiley-VCH Verlag GmbH, 2007.

\bibitem{Rechards2007_SAR}
M.~A. Richards, ``A beginner's guide to interferometric sar concepts and signal
  processing [aess tutorial iv],'' \emph{IEEE Aerospace and Electronic Systems
  Magazine}, vol.~22, no.~9, pp. 5--29, Sept 2007.

\bibitem{Ferretti2007_SARbook}
A.~Ferretti, A.~Monti-Guarnieri, C.~Prati, F.~Rocca, and D.~Massonet,
  \emph{InSAR Principles: Guidelines for SAR Interferometry Processing and
  Interpretation}.\hskip 1em plus 0.5em minus 0.4em\relax ESA Publications,
  2007.

\bibitem{SARsignalmagazine}
G.~Fornaro, F.~Lombardini, A.~Pauciullo, D.~Reale, and F.~Viviani,
  ``Tomographic processing of interferometric {SAR} data: Developments,
  applications, and future research perspectives,'' \emph{IEEE Signal
  Processing Magazine}, vol.~31, no.~4, pp. 41--50, July 2014.

\bibitem{Born1999_principleoptics}
M.~Born and E.~Wolf, \emph{Principles of Optics: Electromagnetic Theory of
  Propagation, Interference and Diffraction of Light}.\hskip 1em plus 0.5em
  minus 0.4em\relax Cambridge University Press, 1999.

\bibitem{Monnier2003_optics}
J.~D. Monnier, ``Optical interferometry in astronomy,'' \emph{Reports on
  Progress in Physics}, vol.~66, no.~5, p. 789, 2003.

\bibitem{Abouraddy11OL}
A.~F. Abouraddy, T.~M. Yarnall, and B.~Saleh, ``Angular and radial mode
  analyzer for optical beams,'' \emph{Opt. Lett.}, vol.~36, pp. 4683--4685,
  2011.

\bibitem{Kurien:14}
B.~G. Kurien, Y.~Rachlin, V.~N. Shah, J.~B. Ashcom, and V.~Tarokh, ``Compressed
  sensing techniques for image reconstruction in optical interferometry,'' in
  \emph{Imaging and Applied Optics}, 2014.

\bibitem{Angleestimate}
Z.~Barber, C.~Harrington, C.~Thiel, W.~Babbitt, and R.~K. Mohan, ``Angle of
  arrival estimation using spectral interferometry,'' \emph{Journal of
  Luminescence}, vol. 130, no.~9, pp. 1614 -- 1618, 2010.

\bibitem{localization}
Y.~Wang, X.~Ma, C.~Chen, and X.~Guan, ``Designing dual-tone radio
  interferometric positioning systems,'' \emph{IEEE Transactions on Signal
  Processing}, vol.~63, no.~6, pp. 1351--1365, March 2015.

\bibitem{imagesignalmagazine}
E.~Thiebaut and J.~F. Giovannelli, ``Image reconstruction in optical
  interferometry,'' \emph{IEEE Signal Processing Magazine}, vol.~27, no.~1, pp.
  97--109, Jan 2010.

\bibitem{Li15SCIrep}
J.~Li, J.~S. Li, Y.~Y. Pan, and R.~Li, ``Compressive optical image
  encryption,'' \emph{Sci. Rep.}, vol.~5, p. 10374, 2015.

\bibitem{Rawat15AP}
N.~Rawat, B.~Kim, I.~Muniraj, G.~Situ, and B.-G. Lee, ``Compressive sensing
  based robust multispectral double-image encryption,'' \emph{Appl. Opt.},
  vol.~54, no.~7, pp. 1782--1793, Mar 2015.

\bibitem{OCTprinciple}
A.~F. Fercher, W.~Drexler, C.~K. Hitzenberger, and T.~Lasser, ``Optical
  coherence tomography--principles and applications,'' \emph{Rep. Prog. Phys.},
  vol.~66, no.~2, pp. 239--303, Jan. 2003.

\bibitem{baha_OCT}
N.~Mohan, I.~Stojanovic, W.~C. Karl, B.~E.~A. Saleh, and M.~C. Teich,
  ``Compressed sensing in optical coherence tomography,'' \emph{Proc. SPIE},
  vol. 7570, 2010.

\bibitem{localizationexample}
M.~Mar\'{o}ti, P.~V\"{o}lgyesi, S.~D\'{o}ra, B.~Kus\'{y}, A.~N\'{a}das,
  A.~L{\'e}deczi, G.~Balogh, and K.~Moln\'{a}r, ``Radio interferometric
  geolocation,'' in \emph{Embedded Networked Sensor Systems}, NY, USA, 2005,
  pp. 1--12.

\bibitem{interferometrystellar}
A.~Glindemann, \emph{Principles of Stellar Interferometry}.\hskip 1em plus
  0.5em minus 0.4em\relax Springer Berlin Heidelberg, 2011.

\bibitem{OCT:book}
M.~E. Brezinski, \emph{Optical Coherenece Tomography}.\hskip 1em plus 0.5em
  minus 0.4em\relax Academic Press, 2006.

\bibitem{OCT_science}
D.~Huang, E.~Swanson, C.~Lin, J.~Schuman, W.~Stinson, W.~Chang, M.~Hee,
  T.~Flotte, K.~Gregory, C.~Puliafito, and a.~et, ``Optical coherence
  tomography,'' \emph{Science}, vol. 254, no. 5035, pp. 1178--1181, 1991.

\bibitem{OCT_retina}
J.~G. Fujimoto, C.~Pitris, S.~A. Boppart, and M.~E. Brezinski, ``Optical
  coherence tomography: An emerging technology for biomedical imaging and
  optical biopsy,'' \emph{Neoplasia}, vol.~2, no. 1–2, pp. 9 -- 25, 2000.

\bibitem{Abouraddy12OL}
A.~F. Abouraddy, T.~M. Yarnall, and B.~E.~A. Saleh, ``Generalized optical
  interferometry for modal analysis in arbitrary degrees of freedom,''
  \emph{Opt. Lett.}, vol.~37, pp. 2889--2891, 2012.

\bibitem{Mardani15OE}
D.~Mardani, A.~F. Abouraddy, and G.~K. Atia, ``Efficient modal analysis using
  compressive optical interferometry,'' \emph{Opt. Express}, vol.~23, no.~22,
  pp. 28\,449--28\,458, Nov 2015.

\bibitem{Wang12NP}
J.~Wang, J.~Yang, I.~M. Fazal, N.~Ahmed, Y.~Yan, H.~Huang, Y.~Ren, Y.~Yue,
  S.~Dolinar, M.~Tur, and A.~E. Willner, ``Terabit free-space data transmission
  employing orbital angular momentum multiplexing,'' \emph{Nat. Photon.},
  vol.~6, pp. 488--496, January 2012.

\bibitem{Fienup_phase}
J.~R. Fienup, ``Phase retrieval algorithms: a comparison,'' \emph{Appl. Opt.},
  vol.~21, no.~15, pp. 2758--2769, Aug 1982.

\bibitem{Eldar_phase}
D.~Kogan, Y.~C. Eldar, and D.~Oron, ``On the 2d phase retrieval problem,''
  \emph{IEEE Trans. on Signal Process.}, vol.~65, no.~4, pp. 1058--1067, Feb
  2017.

\bibitem{Pal_ICASSP}
H.~Qiao and P.~Pal, ``Sparse phase retrieval with near minimal measurements: A
  structured sampling based approach,'' in \emph{IEEE Int. Conf. on Acoustics,
  Speech and Signal Processing}, March 2016, pp. 4722--4726.

\bibitem{Pal_GLOBAL}
H.~Qiao and P.~Pal, ``Sparse phase retrieval using partial nested fourier
  samplers,'' in \emph{IEEE GlobalSIP}, Dec 2015, pp. 522--526.

\bibitem{Phaselift}
E.~J. Candes, T.~Strohmer, and V.~Voroninski, ``Phaselift: Exact and stable
  signal recovery from magnitude measurements via convex programming,''
  \emph{Communications on Pure and Applied Mathematics}, vol.~66, no.~8, pp.
  1241--1274, 2013.

\bibitem{Mirhosseini14PRL}
M.~Mirhosseini, O.~S. Maga{\~n}a-Loaiza, S.~M.~H. Rafsanjani, and R.~W. Boyd,
  ``Compressive direct measurement of the quantum wave function,'' \emph{Phys.
  Rev. Lett.}, vol. 113, p. 090402, 2014.

\bibitem{Howland14PRL}
G.~A. Howland, J.~Schneeloch, D.~J. Lum, and J.~C. Howell, ``Simultaneous
  measurement of complementary observables with compressive sensing,''
  \emph{Phys. Rev. Lett.}, vol. 112, p. 253602, 2014.

\bibitem{Duarte08SPM}
M.~F. Duarte, M.~A. Davenport, D.~Takhar, J.~N. Laska, T.~Sun, K.~F. Kelly, and
  R.~G. Baraniuk, ``Single-pixel imaging via compressive sampling,'' \emph{IEEE
  Signal Process. Mag.}, vol.~25, pp. 83--91, 2008.

\bibitem{Martinez-Leon:17}
L.~Mart\'{i}nez-Le\'{o}n, P.~Clemente, Y.~Mori, V.~Climent, J.~Lancis, and
  E.~Tajahuerce, ``Single-pixel digital holography with phase-encoded
  illumination,'' \emph{Opt. Express}, vol.~25, no.~5, pp. 4975--4984, Mar
  2017.

\bibitem{Clemente:13}
P.~Clemente, V.~Dur\'{a}n, E.~Tajahuerce, P.~Andr\'{e}s, V.~Climent, and
  J.~Lancis, ``Compressive holography with a single-pixel detector,''
  \emph{Opt. Lett.}, vol.~38, no.~14, pp. 2524--2527, Jul 2013.

\bibitem{Liu:10}
X.~Liu and J.~U. Kang, ``Compressive {SD-OCT}: the application of compressed
  sensing in spectral domain optical coherence tomography,'' \emph{Opt.
  Express}, vol.~18, no.~21, pp. 22\,010--22\,019, Oct 2010.

\bibitem{Candes08CR}
E.~J. Candes, ``The restricted isometry property and its implications for
  compressed sensing,'' \emph{Compte Rendus de l'Academie des Sciences}, vol.
  346, pp. 589--592, 2008.

\bibitem{candesRIPless}
E.~J. Candes and Y.~Plan, ``A probabilistic and {RIPless} theory of compressed
  sensing,'' \emph{IEEE Trans. Inf. Theory}, vol.~57, no.~11, pp. 7235--7254,
  Nov 2011.

\bibitem{baha}
B.~E.~A. Saleh and M.~C. Teich, \emph{Fundamentals of Photonics}.\hskip 1em
  plus 0.5em minus 0.4em\relax John Wiley \& Sons, 1994.

\bibitem{NAMIAS80JAM1}
V.~Namias, ``The fractional {F}ourier transform and its application in quantum
  mechanics,'' \emph{IMA J. Appl. Math.}, vol.~25, pp. 241--256, 1980.

\bibitem{Pei07TSP}
S.~C. Pei and J.~J. Ding, ``Relations between gabor transforms and fractional
  fourier transforms and their applications for signal processing,'' \emph{IEEE
  Trans. Signal Process.}, vol.~55, no.~10, pp. 4839--4850, Oct 2007.

\bibitem{NAMIAS80JAM2}
V.~Namias, ``Fractionalization of hankel transforms,'' \emph{IMA J. Appl.
  Math.}, vol.~26, pp. 187--197, 1980.

\bibitem{Bozinovic13Sc}
N.~Bozinovic, Y.~Yue, Y.~Ren, M.~Tur, P.~Kristensen, H.~Huang, A.~Willner, and
  S.~Ramachandran, ``Terabit-scale orbital angular momentum mode division
  multiplexing in fibers,'' \emph{Science}, vol. 340, pp. 1545--1548, 2013.

\bibitem{siegman1986lasers}
A.~Siegman, \emph{Lasers}.\hskip 1em plus 0.5em minus 0.4em\relax University
  Science Books, 1986.

\bibitem{Ozaktas01wiley}
H.~M. Ozaktas, Z.~Zalevsky, and M.~A. Kutay, \emph{The Fractional {F}ourier
  Transform with Applications in Optics and Signal Processing}.\hskip 1em plus
  0.5em minus 0.4em\relax Wiley, 2001.

\bibitem{sci_rep2017}
L.~Martin, D.~Mardani, H.~E. Kondakci, W.~D. Larson, S.~Shabahang, A.~K.
  Jahromi, T.~Malhotra, A.~N. Vamivakas, G.~K. Atia, and A.~F. Abouraddy,
  ``Basis-neutral {Hilbert-space} analyzers,'' \emph{Scientific Reports},
  vol.~7, 2017.

\bibitem{BP}
S.~S. Chen, D.~L. Donoho, and M.~A. Saunders, ``Atomic decomposition by basis
  pursuit,'' \emph{SIAM Journal on Scientific Computing}, vol.~20, no.~1, pp.
  33--61, 1998.

\bibitem{Simpleproof}
R.~Baraniuk, M.~Davenport, R.~DeVore, and M.~Wakin, ``A simple proof of the
  restricted isometry property for random matrices,'' \emph{Constructive
  Approximation}, vol.~28, no.~3, pp. 253--263, 2008.

\bibitem{Cauchy}
J.~M. Steele, \emph{The Cauchy-Schwarz Master Class: An Introduction to the Art
  of Mathematical Inequalities}.\hskip 1em plus 0.5em minus 0.4em\relax New
  York, NY, USA: Cambridge University Press, 2004.

\bibitem{Hoeffding}
W.~Hoeffding, ``Probability inequalities for sums of bounded random
  variables,'' \emph{Journal of the American Statistical Association}, vol.~58,
  no. 301, pp. 13--30, 1963.

\bibitem{RIPless}
E.~J. Candes and Y.~Plan, ``A probabilistic and {RIPless} theory of compressed
  sensing,'' \emph{IEEE Transactions on Information Theory}, vol.~57, no.~11,
  pp. 7235--7254, Nov 2011.

\bibitem{LASSO}
R.~Tibshirani, ``Regression shrinkage and selection via the lasso,''
  \emph{Journal of the Royal Statistical Society, Series B}, vol.~58, pp.
  267--288, 1994.

\bibitem{Dantzig}
E.~Candes and T.~Tao, ``The dantzig selector: Statistical estimation when p is
  much larger than n,'' \emph{The Annals of Statistics}, vol.~35, no.~6, pp.
  2313--2351, 2007.

\end{thebibliography}

\end{document}